\documentclass{article}
\usepackage{amsmath,framed}
\usepackage{amsfonts}
\usepackage{amssymb}
\usepackage{color}
\usepackage{graphicx}
\newtheorem{theorem}{Theorem}

\newtheorem{remark}[theorem]{Remark}

\numberwithin{theorem}{section}
\newenvironment{proof}[1][Proof]{\noindent\textbf{#1.} }{\ \rule{0.5em}{0.5em}}
\newcommand{\bfi}{\bfseries\itshape}
\newcommand{\bz}{{\bf z}}
\newcommand{\rd}{{\rm d}}
\newcommand{\bX}{{\boldsymbol{X}}}
\newcommand{\rem}[1]{}
\newcommand{\pp}[2]{\frac{\partial #1}{\partial #2}}
\newcommand{\dd}[2]{\frac{d #1}{d #2}}
\newcommand{\dede}[2]{\frac{\delta #1}{\delta #2}}
\newcommand{\pa}{{\partial}}
\newcommand{\papa}[2]{\frac{\partial #1}{\partial #2}}

\newcommand{\bzeta}{\boldsymbol{\zeta}}

\newcommand{\ad}{\textrm{\large ad}}
\newcommand{\oiint}{\int\hspace{-.37cm}\text{\large$\bigcirc$}\hspace{-.38cm}\int}

\newcommand{\comment}[1]{\vspace{1 mm}\par
\marginpar{\large\underline{}}\noindent
\framebox{\begin{minipage}[c]{0.95 \textwidth} \color{blue}\bfi #1
\end{minipage}}\vspace{2 mm}\par}

\def\contract{\makebox[1.2em][c]{\mbox{\rule{.6em}
{.01truein}\rule{.01truein}{.6em}}}}
\begin{document}

\title{
The geodesic Vlasov equation\\and its integrable moment closures}
%
\author{
\hspace{-1.1cm}
Darryl D. Holm$^{1,\,2}$, Cesare Tronci$^{3}\!$\\
{\hspace{-1.2cm}\footnotesize $^1$ \it Department of Mathematics, Imperial College London, London SW7 2AZ, UK}\\
{\hspace{-1.2cm}\footnotesize $^2$ \it Institute for Mathematical
Sciences,
Imperial College London, 53 Prince's Gate, London SW7 2PG, UK} \\
{\hspace{-1.2cm}\footnotesize $^3$\,\it Section de Math\'ematiques, \'Ecole Polytechnique F\'ed\'erale de Lausanne, CH-1015 Lausanne, Switzerland}
\\ \\
}
\date{
\normalsize\today}
\maketitle

\tableofcontents
\newpage

\begin{abstract}
Various integrable geodesic flows on Lie groups are shown to arise by taking moments of a geodesic Vlasov equation on the group of canonical transformations. This was already known for both the one- and two-component Camassa-Holm systems \cite{GiHoTr2005,GiHoTr2007}. The present paper extends our earlier work to recover another integrable system of ODE's that was recently introduced by Bloch and Iserles \cite{BlIs2006}.
Solutions of the Bloch-Iserles system are found to arise from the Klimontovich solution of the geodesic Vlasov equation. These
solutions are shown to form one of the legs of a dual pair of
momentum maps.%
\rem{This dual pair recovers both the Legendre-to-Euler map and the Kelvin-Noether theorem, thereby illustrating two profound geometric aspects of kinetic theory.}
The Lie-Poisson structures for the dynamics of truncated moment hierarchies are also presented in this context.
\end{abstract}

\bigskip

\section{Introduction}

Kinetic equations govern the evolution of probability distributions
in the phase space of many-particle systems in non-equilibrium
statistical mechanics. For example, the phase-space probability
distribution of a many-particle system whose correlations are
negligible is governed by the {\it collisionless Boltzmann}
equation, also known as {\it Vlasov equation} \cite{Vl1961}. This
equation encodes evolution of the Vlasov single-particle probability
distribution $f({\bf q},{\bf p},t)$ as conservation along phase
space trajectories, written as
\begin{equation}\label{VlasovChar}
\frac{df}{dt}=0 \qquad\text{along}\qquad \dot{\bf q}=\frac{\pa
h}{\pa \bf p}\,, \quad \dot{\bf p}=-\,\frac{\pa h}{\pa \bf q}\,,
\end{equation}
where $h$ is the single-particle Hamiltonian, often expressed as the sum of kinetic and potential energies
\[
h=\frac{|{\bf p}|^2}{2}+V({\bf q})
\]
in physical applications. Applying the canonical Poisson bracket $\{\cdot,\,\cdot\}$ in the phase space variables ($\bf q,p$) expresses the Vlasov equation in its familiar form
\begin{equation}\label{Vlasov}
\pp{f}{t} +\big\{f,h\big\}=0 \,.
\end{equation}
The Hamiltonian structure of this system is well known \cite{MaWe1981}. Namely, the Vlasov equation possesses a Lie-Poisson bracket defined on the Lie algebra of canonical transformations, such that the Liouville theorem for preservation of the volume on phase space  entirely characterizes Vlasov dynamics.

The moment method, widely used in probability theory, provides approximate descriptions of the Vlasov solutions. Moments are functionals of the distribution function $f$ obtained by projections onto  the space of phase space polynomials (symmetric tensors).
Since the Vlasov distribution depends on both position {\bf q} and momentum {\bf p}, one may define two
different types of moments. These are the {\it kinetic moments} and the {\it statistical moments}. Kinetic moments are given by projection of the
Vlasov distribution $f({\bf q,p})$ onto the $n$-th power of the single-particle momentum, $\bf p$. In contrast, statistical moments are integrals of $f({\bf q,p})$ against the $n$-th power of the full phase space vector, ${\bf z}=({\bf q,p})$. The remarkable property of these two hierarchies of moment projections is that they each define equivariant momentum maps \cite{HoLySc1990,ScWe1994,GiHoTr2008}. Consequently, the resulting moment dynamics is again Lie-Poisson. Moment equations possess interesting closures, which are given by the particular Lie algebra structure determining their Lie-Poisson bracket. For example, kinetic moments of the Vlasov equation at zero-th and first order yield the familiar closure known as \emph{ideal fluid dynamics}.

Remarkably, these kinetic moment equations are associated to a family of  integrable dynamical systems, whose most famous example is probably the Benney system for shallow water
dynamics \cite{Be1973,Gi1981}. This convergence of different areas of mathematical physics also occurs for several other integrable equations. For example, as shown in \cite{GiHoTr2005}, a specific form of the first-order kinetic moment equation yields the Camassa-Holm equation \cite{CaHo1993}. Extending the system to include the zero-th order moment yields another integrable system; the two-component
Camassa-Holm system \cite{GiHoTr2007}. Interestingly enough, these
Camassa-Holm systems are geodesic flows on \emph{different} Lie groups, arising as moment closures of the \emph{same} kinetic equation, called here the {\bfi geodesic Vlasov equation} or EP${Can}$. The latter acronym refers to the Euler-Poincar\'e (EP) equation on the group of canonical transformations $Can$ acting on phase space $T^*Q$.%
\footnote{The group of canonical transformations $Can$ is also known as the Hamiltonian diffeomorphisms $Diff_{\rm Ham}$.}
A special case of EP${Can}$ for canonical transformations whose generating functions are \emph{linear} in the canonical momentum has recently appeared in the theory of \emph{metamorphoses} in imaging science \cite{HoTrYo2007}.

This paper reviews the theory of geodesic equations on the statistical moments and shows how such equations possess an additional  interesting closure, which is related to the space of purely quadratic Hamiltonian functions.
We find that such a closure yields a particular case of yet another
integrable system, recently discovered by Bloch and Iserles
\cite{BlIsMaRa2005,BlIs2006}. Moreover, extending to inhomogeneous
quadratic Hamiltonians yields complete equivalence between moment
equations and the Bloch-Iserles (BI) system.

\paragraph{Plan} The rest of this section adds a few more remarks about the Lie-Poisson bracket for the Vlasov equation. Section 2 is devoted to the Hamiltonian structure of the Vlasov moments and their truncations. Section 3 formulates the geodesic Vlasov equation, presents its dual pair and illustrates the geometric footing of kinetic theory. Section 4 shows how both of the Camassa-Holm systems are obtained as geodesic equations on kinetic moments. The last section derives the BI system from the statistical moment equations and presents the corresponding Klimontovich solutions.

\subsection{The Vlasov kinetic equation}

The Vlasov equation is a Lie-Poisson Hamiltonian system on the group of canonical transformations of the phase space $T^*Q$ for a configuration manifold $Q$ \cite{MaWe1981}. The dynamics of Lie-Poisson systems takes place on the dual $\mathfrak{g}^*$ of the Lie algebra $\mathfrak{g}$ of the symmetry group $G$.
In this case $G={Can}(T^*Q)$ and
$\mathfrak{g}=\mathfrak{X}_{\rm can}(T^*Q)$. That is, the Lie algebra is the infinite-dimensional space of Hamiltonian vector fields. Given
the Lie algebra isomorphism between Hamiltonian vector fields and
phase-space functions ($\mathfrak{X}_{\rm can}\simeq\mathcal{F}$),
the dynamical variable is a phase-space distribution $f({\bf q,p})$,
i.e., a density on phase space ($f\in\mathcal{F}^*\!\simeq{\rm
Den}$). Upon using the definition of canonical Poisson bracket
$\left\{\cdot,\cdot\right\}$, the Vlasov Lie-Poisson structure is
found to be
\begin{equation}
\{F,H\}[f]=\iint f ({\bf q,p})\left\{\frac{\delta F}{\delta
f},\frac{\delta H}{\delta f}\right\}\,{\rm d}^{K} {\bf q}\ {\rm
d}^{K} {\bf p}
\end{equation}
where $K={\rm dim}(Q)$. The Vlasov equation (\ref{Vlasov}) is
recovered upon choosing $F=f$ and $h=\delta H/\delta f$.

In many physical applications, the Vlasov Hamiltonian is the sum of
kinetic and potential energy. For example, electrostatic or gravitational interactions are governed in the absence of collisions by the
Poisson-Vlasov system whose Hamiltonian is given by
\begin{equation}\label{VlasovPoisson}
H[f]=\iint\! f({\bf q,p})\left(\frac12\,|{\bf
p}|^2+\Delta^{-1\!}\!\int\! f({\bf q,p'})\ {\rm d}^{K\,}{\bf
p}'\right) {\rm d}^K{\bf q}\ {\rm d}^K{\bf p} \,,
\end{equation}
where $\Delta^{-1}$ denotes convolution with the Green's function of
the Laplace operator.

\section{Hamiltonian structure of Vlasov moments}
The moment method is a popular approach in kinetic systems theory.
This approach is justified geometrically because taking moments of
the Vlasov distribution is a momentum map
\cite{HoLySc1990,ScWe1994,GiHoTr2008}. This momentum map arises via
the dual of a Lie algebra homomorphism arising from the well-known
isomorphism between symmetric tensors and polynomials. The main
point is that this momentum map endows the space of symmetric
tensors with a Lie bracket, thereby generating a well defined Lie
algebra. In what follows, we shall analyze the cases of kinetic and
statistical moments separately and then discuss their similarities.

\subsection{Kinetic moments and the Schouten concomitant}

{\bfi Kinetic moments} are constructed from the following fiber integral \cite{QiTa2004}
\begin{multline}\label{FibIntMom}
A_n({\bf q}, t):=\int_{T^*_{\bf q}Q\!}\left({\bf p}\cdot{\rm d}{\bf
q}\right)^n f({\bf q,p},t)\ {\rm d}^K{\bf q}\wedge{\rm d}^K{\bf p}
\\
= \sum_{i_1\dots i_n=1}^K
\int_{T^*_{\bf q}Q}  p_{\,i_1}\dots p_{\,i_n} \
{\rm d}q^{i_1}\otimes\dots\otimes
{\rm d}q^{i_n}\,f({\bf q, p}, t)\ {\rm d}^K{\bf
q}\wedge{\rm d}^K{\bf p}
\\
=\sum_{i_1\dots i_n=1}^K\!
\big(A_n({\bf q}, t)\big)_{i_1\dots i_n\,} \
{\rm d}q^{i_1}\otimes\dots\otimes{\rm d}q^{i_n}
\otimes{\rm d}^K{\bf q} \,,
\end{multline}
where ${\bf p}\cdot{\rm d}{\bf q}$ denotes the canonical one form
(canonical momentum) and ${\rm d}^K{\bf q}$ is the volume element on the configuration space $Q$. This construction projects the Vlasov
distribution onto the space of symmetric tensors. In particular,
kinetic moments are defined as symmetric covariant tensor fields
carrying the volume element. That is, they are {\it symmetric covariant
tensor densities}.

The moments are functionals of the Vlasov density $f$. Hence, their variational derivative may be computed by applying the chain rule as
\begin{align*}
\frac{\delta F}{\delta f}
= \sum_{n=0}^\infty\,\frac{\delta F}{\delta
A_n}
\contract
\frac{\delta
A_n}{\delta f}
&
:=\sum_{n=0}^\infty\,\sum_{\,i_1\dots i_n=1}^K\!\frac{\delta
(A_n)_{i_1\dots i_n}}{\delta f}\,\frac{\delta F}{\delta
(A_n)_{i_1\dots i_n}}
\\&
=\sum_{n=0}^\infty\,\sum_{\,i_1\dots i_n=1}^K
p_{\,i_1}\dots p_{\,i_n}
\,\frac{\delta F}{\delta
(A_n)_{i_1\dots i_n}}
\\&
=: \sum_{n=0}^\infty\, \frac{\delta F}{\delta A_n} \contract\, {\bf
p}^{n\!}
\,,
\end{align*}
which explicitly defines the contraction operation $\contract$.
This chain rule formula expresses the Lie algebra homomorphism (isomorphism) from symmetric tensors to polynomials, whose dual is the momentum map associated to the moments \cite{GiHoTr2008}. Inserting the chain rule formula into the Vlasov bracket yields the Lie-Poisson bracket for moments,
\begin{equation}
\left\{F,G\right\}[A]=- \!\sum_{n,m=0\,}^\infty
\int\! A_{m+n-1}({\bf q})\,\contract\!\left[ \frac{\delta F}{\delta
A_n},\frac{\delta G}{\delta A_m}\right]
\,{\rm d}^3{\bf q}
\label{moment-LPB}
\end{equation}
in which the bracket%
\footnote{The operator $\mathcal{S}$ here takes the symmetric part of its argument.}
\begin{equation}\label{MomSchouten}
\Bigg[ \frac{\delta F}{\delta A_n}\,,\,\frac{\delta G}{\delta
A_m}\Bigg]=\, \text{\Large$\mathcal{S}$}\!\left(n
\,\Bigg(\frac{\delta F}{\delta
A_n}\cdot\nabla\Bigg)\otimes\frac{\delta G}{\delta A_m} \, -\,
m\, \Bigg(\frac{\delta G}{\delta
A_m}\cdot\nabla\Bigg)\otimes\frac{\delta F}{\delta A_n}
\right)
\end{equation}
is inherited from the canonical Poisson bracket.
Here the notation $A\cdot B$ for one-index contraction between covariant and contravariant tensors is written as $(A\cdot B)_{ij...}^{hl...} = A_{ij...k}B^{khl...}$ and analogously for $B\cdot A=(B\cdot A)^{km...}_{jl...} =
B^{km...i}A_{ijl...}$.
This bracket
is well known in differential geometry as an invariant differential
operator of first order \cite{Ni1955}. In fact, this operation is a Lie
bracket, which is known as the {\it Schouten concomitant} or {\it
symmetric Schouten bracket}. See, e.g.,
\cite{GiHoTr2008} for more discussions and references.

\begin{remark}[History of Lie-Poisson structure for kinetic moments]$\quad$\\
In one dimension, the moment Lie-Poisson structure (\ref{moment-LPB}) is the
Kupershmidt-Manin bracket \cite{KuMa1978} which was found in the
context of the integrable Benney system in shallow water theory.
Lebedev was the first to establish its relation with the Lie algebra
of Hamiltonian vector fields in \cite{Le1979} and Gibbons recognized
later \cite{Gi1981} its direct relation to the Vlasov flow. In
higher dimensions, Kupershmidt introduced a multi-index notation
\cite{Ku1987}, corresponding to the occupation number representation
of the symmetric Schouten bracket. This observation suggested the quantum-like framework for kinetic moments in \cite{GiHoTr2008}, where the moment space is
described in terms of a bosonic Fock space.
\end{remark}

The moment algebra comprises symmetric contravariant tensor fields and these may be characterized as the Fock space
represented by a direct sum of symmetric powers of vector fields given by
\begin{equation}
\mathfrak{g}:=\bigoplus_{n=0}^\infty\left(\,\bigvee_{i=0}^{n}\,\mathfrak{X}(Q)\,\right)
\qquad\text{with}\quad
\bigvee_{i=0}^{n}\mathfrak{X}
:=\,\text{\large$\mathcal{S}$}\!\left(\bigotimes_{i=0}^n\,\mathfrak{X}\right)
=:\mathfrak{g}_n
\,.
\label{V-devil-def}
\end{equation}
This is reminiscent of the {\it universal enveloping algebra of the
diffeomorphism group} Diff($Q$), which is the enveloping algebra
$\mathcal{U}(\mathfrak{X})$ of vector fields $\mathfrak{X}(Q)$ on
the configuration space $Q$. It is a standard result that the graded
structure of an enveloping algebra possesses a
 Poisson bracket structure \cite{DaSWe1999}.

\begin{remark}[Kinetic moments and Poisson-Lie groups]$\,$\\
Interestingly enough, the Schouten concomitant identifies the kinetic moment algebra with the Lie algebra of symbols of differential operators. This identification is quite suggestive, since differential symbols are known to be a subalgebra of the Lie algebra of pseudo-differential symbols \cite{KhZa1995}. While the differential symbols (the moment algebra) are supposed to have no underlying Lie group structure, the group ${\Psi}D$ of pseudo-differential operators is a well defined Poisson-Lie group \cite{KhZa1995}. This suggests that the characterization of coadjoint orbits for moment dynamics requires the complete Poisson-Lie group structure of pseudo-differential operators. A similar direction involving vector fields was followed by Ovsienko and Roger \cite{OvRo1999}. Also, the appearance of the Wick-ordered product from quantum theory in this Poisson-Lie group context \cite{KhZa1995,OvRo1999} implies a further relation to moment dynamics, whose quantum-like
creation and annihilation operators were presented in \cite{GiHoTr2008}.
\end{remark}

The moment algebra carries a
graded structure
($\mathfrak{g}=\text{\large$\oplus$}_i\,\mathfrak{g}_i$) with
filtration
\begin{equation}\label{KinMomFiltr}
\big[\,\mathfrak{g}_{n\,},\mathfrak{g}_m\big]\subseteq
\mathfrak{g}_{n+m-1}
\end{equation}
thereby recovering the space of vector fields
$\mathfrak{g}_1=\mathfrak{X}$ as a particular subalgebra. The
\emph{largest} subalgebra is however
$\mathfrak{g}_{0\,}\text{\large$\oplus$}\,\mathfrak{g}_1\simeq\mathfrak{X}\,\circledS\,\mathcal{F}$,
i.e. the semidirect product of vector fields with scalar functions.
This space appears in the description of ideal
compressible fluids, where the $\mathfrak{X}$-variable is the fluid
velocity and the $\mathcal{F}$-variable is associated to the fluid
density. Geodesic motion on this space has also recently appeared in the \emph{metamorphosis} process in imaging science \cite{HoTrYo2007}.

The moment equations may be written in Lie-Poisson form as
\begin{equation}
\pp{A_m}{t}=\,-\!\sum_{m,n=0}^\infty
\ad^*_\text{\!\small$\dede{H}{A_n}$}\, A_{n+m-1}
\end{equation}
where ad$^*$ is the Lie algebraic coadjoint operation defined using the pairing
\begin{equation*} \sum_{k,n=0}^\infty\big\langle\,
\ad^*_{\beta_n} \,A_{k\,},\,\alpha_{k-n+1}\big\rangle:=
\sum_{k,n=0}^\infty\big\langle\,
A_{k\,},\,\left[\beta_n,\,\alpha_{k-n+1}\right]\big\rangle \,.
\end{equation*}
The explicit expression for ad$^*$ is given in \cite{GiHoTr2008} in
any number of dimensions. Here we present the one-dimensional
case which will be needed in the following sections. The Schouten
concomitant in equation (\ref{MomSchouten}) assumes a particularly simple form in 1D
\begin{equation}
\left[\alpha_m,\beta_n\right]\,=\,m\,\alpha_n\,\partial_q\beta_n\,-\,n\,\beta_n\,\partial_q\alpha_m
\end{equation}
where $q$ is the spatial variable. Simple use of integration by
parts yields the following covariant tensor density of rank $k-n+1$:
\begin{equation}
\ad_{\beta_{n}}^{\ast}A_{k}=\left(  k+1\right) \,
A_{k}\,\partial_q \beta_{n} + n\,\beta_{n}\,\partial_q A_{k}
\,.
\label{Kirillov-ad-action}
\end{equation}
This operation was introduced by Kirillov \cite{Ki1982}, who first
envisioned the possibility of a Lie-Poisson bracket on the symmetric
Schouten algebra. Familiar versions of this operator with $n=1$ or
$k=n$ arise in the theory of ideal fluid dynamics, soliton dynamics
and image matching, while very little is known for other values on
$n,k$. Some features of this intriguing open question are
investigated further below, in dealing with truncations of moment
hierarchies.

\subsection{Statistical moments and their Lie-Poisson structure}

As we have seen, the fiber integral defining the {\it kinetic moment}
hierarchy in (\ref{FibIntMom}) requires a kinetic equation on a
cotangent bundle. In contrast, the notion of {\it statistical moments}
is given on a symplectic vector space. Upon denoting $\bf z=(q,p)$,
the definition of the $n$-th statistical moment is given by
\begin{equation}
X^n(t):=\int {\bf z}^n f({\bf z},t)\ {\rm d}^N{\bf z}
\end{equation}
where the upper index $n$ in the integrand denotes tensor power
(${\bf z}^n=\otimes^n{\bf z}$), while for the time-dependent tensors
$X^n(t)$ it denotes the tensor rank. This definition places the statistical moments  and kinetic moments into the same mathematical  framework. The first observation is that
statistical moments are symmetric contravariant tensors on phase
space, which is now a symplectic vector space $V$ of even dimension
$N=2K$ (eventually $V=\mathbb{R}^N$) with elements $\bz=z^i {\bf
e}_i\in V$.

The moment Poisson bracket for statistical moments may be obtained by following exactly the same steps as in the previous discussion for kinetic moments. That is, one inserts the chain rule formula
\begin{eqnarray}
\frac{\delta F}{\delta f}
&=&
\sum_{n=0}^\infty\,
\left(\pp{F}{X^n}\right)_{i_1\dots i_n}\left(\dede{X^n}{f}\right)^{i_1\dots
i_n}
\nonumber\\
&:=& \sum_{n=0}^\infty\ \pp{F}{X^n}\contract\dede{X^n}{f}
=
\,\sum_{n=0}^\infty\, \pp{F}{X^n}\contract \,\bz^{n\!\!}
\label{contract-def}
\end{eqnarray}
(and definition of $\contract$) into the Vlasov structure, which may then be written as
\begin{align*}
\{F,G\}[f]=\iint\! f(\bz)\left\{\frac{\delta F}{\delta
f},\frac{\delta G}{\delta f}\right\}\, {\rm d}^N\bz =\iint\!
f(\bz)\left[\, \mathbb{J}\contract \left( \frac{\partial}{\partial
\bz}\frac{\delta F}{\delta f}\otimes\frac{\partial}{\partial
\bz}\frac{\delta G}{\delta f}\right)\right] {\rm d}^N\bz
\end{align*}
where $\mathbb{J}$ is a non-degenerate two form. That is, $\mathbb{J}$ is a $N\times N$
antisymmetric matrix of maximal rank. Since equation (\ref{contract-def}) implies
\begin{equation}
\frac{\partial}{\partial \bz}\frac{\delta F}{\delta f} = \sum_n\,
n\,\frac{\partial F}{\partial X^n}\contract\,\bz^{n-1}\,,
\end{equation}
it follows that the moment Poisson structure is
\begin{align}\nonumber
\{F,G\}(X)\,&=\!\sum_{n,m=0}^\infty
X^{n+m-2}\contract\left[\frac{\partial F}{\partial
X^n},\frac{\partial G}{\partial X^m}\right]
\\
&=:\!\sum_{n,m=0}^\infty \left\langle
X^{n+m-2\,},\left[\frac{\partial F}{\partial X^n},\frac{\partial
G}{\partial X^m}\right]\right\rangle
\end{align}
where
\begin{equation}\label{mombrkt}
\left[\frac{\partial F}{\partial X^n},\frac{\partial G}{\partial
X^m}\right]:=n\,m\
\text{\large$\mathcal{S}$}\!\left(\frac{\partial F}{\partial
X^n}\cdot \mathbb{J}\cdot \frac{\partial G}{\partial
X^m}\right)
\end{equation}
is the moment Lie bracket, in which again $\mathcal{S}$ operates to
take the symmetric part of its argument. Recall that $\mathbb{J}$ is
considered as a contravariant antisymmetric matrix, i.e. it
possesses upper indexes $\mathbb{J}^{ij}=-\mathbb{J}^{ji}$.

Thus, again, the isomorphism between symmetric tensors and
polynomials produces the momentum map associated with the moments \cite{ScWe1994}.
 In turn, this means that the Lie-Poisson bracket for statistical moments is inherited from the Vlasov Lie-Poisson structure. In contrast to the Schouten concomitant (\ref{MomSchouten}) for kinetic moments, the Lie bracket for statistical moments in (\ref{mombrkt}) still involves the symplectic matrix $\mathbb{J}$ (Poisson tensor). Thus, the dynamics of the statistical moments depends explicitly on the original symplectic structure. This allows, for example, the direct construction of moment invariants (Casimirs) as presented in \cite{HoLySc1990}. Also, the moment
algebra (\ref{mombrkt}) involves symmetric tensors that are {\it
covariant}, rather than contravariant as happens for the Schouten
concomitant.

\begin{remark}[Statistical moments and accelerator beam optics] $\quad$\\
Statistical moments are important, for example, in the
study of beam dynamics in particle accelerators. In this
framework, they are defined  as \cite{Ch1983}
\begin{multline*}
\mathcal{M}_n^{\widehat{n}}(t):=\iint {\bf p}^n\, {\bf
q}^{\widehat{n}}\,f({\bf q,p},t)\ {\rm d}^p{\bf q}\ {\rm d}^p{\bf q}
\\
=\left(\mathcal{M}_n^{\widehat{n}}(t)\right)_{j_1...j_n}^{i_1...i_{\widehat{n}}}
{\bf e}_{i_1}\otimes...\otimes{\bf e}_{i_{\widehat{n}}}\otimes {\bf
e}^{j_1}\otimes...\otimes{\bf e}^{j_n}
\end{multline*}
where ${\bf e}_k$ is a basis element of the configuration vector
space $Q$, while ${\bf e}^k$ is its dual (so that $V=Q\times Q^*$).
In 1D, the {\it beam emittance}
\[\epsilon:=\left(\mathcal{M}_0^2\,\mathcal{M}_2^0-(\mathcal{M}_1^1)^2\right)^{1/2}\,,\]
known as the Courant-Snyder invariant \cite{CoSn1958},
was recognized as a moment invariant (Casimir).  This observation led to the study of more general moment invariants \cite{Dr1990,HoLySc1990}. A geometric investigation of statistical moments was carried
out in \cite{ScWe1994},  where the moment algebra was related to the
Heisenberg algebra on phase space. For particle
accelerator design, moments are often used in computational efforts to account for space charge effects and other beam-related problems.
\end{remark}

As for the kinematic moments, one characterizes the Lie algebra of
statistical moments by using the grading,
\begin{equation}
\mathfrak{g}:=\bigoplus_{n=0}^\infty\left(\,\bigvee_{i=0}^{n}\,V^*\!\right)=:\bigoplus_{i=0}^\infty\,\mathfrak{g}_i
\end{equation}
with the filtration
\begin{equation}\label{StatMomFiltr}
\big[\,\mathfrak{g}_{n\,},\mathfrak{g}_m\big]\subseteq
\mathfrak{g}_{n+m-2}
\end{equation}
which shows how symmetric matrices $\mathfrak{g}_2=V^*\vee V^*={\rm
Sym}^*(N)$ form a particular subalgebra (here we denote by ${\rm
Sym}^*(N)$ covariant symmetric matrices). The largest subalgebra is
given by $\mathfrak{g}_0\,\text{\large$\oplus$}\,\mathfrak{g}_1\,\text{\large$\oplus$}\,\mathfrak{g}_2=\mathbb{R}\,\text{\large$\oplus$}\,
V^*\text{\large$\oplus$}\, {\rm Sym}^*(N)$ and it will play a central
role in the remainder of this paper.

\begin{remark}[Occupation number representation]$\quad$\\
Analogously to the Lie algebra of kinetic moments, the statistical
moments also carry a bosonic Fock space structure, where appropriate
occupation numbers may be defined through the introduction of the
multi-index notation
$\bz^\sigma:=(z^1)^{\sigma_1}\dots(z^N)^{\sigma_N}$. Indeed, taking
moments by $X^\sigma(t)=\int \bz^\sigma f(\bz,t)\,{\rm d}^N\bz$
yields the occupation number representation for statistical Vlasov
moments.
\end{remark}

The moment equations are written in terms of the Lie algebra
coadjoint operator $\rm ad^*$ as
\begin{align}\nonumber
\frac{dX^m}{dt} &= -\sum_{n=0}^\infty \textrm{\large
ad}^*_\text{\!\small$\pp{H}{X^n}$}\, X^{n+m-2}
\\
&= -\,m\sum_{n=0}^\infty n\ \text{\large$\mathcal{S}$}\left(
\left(\pp{H}{X^n}\cdot\mathbb{J}\right)\!\contract\,
X^{m+n-2}\right) \label{Vlasov-moment-eqn}
\end{align}
so that
\begin{equation*}
\Big(\dot{X}^m\Big)^{\!i_1...i_m}= -\,m\sum_{n=0}^\infty n\
\text{\large$\mathcal{S}$}\!\left(
\left(\pp{H}{X^n}\right)_{\!\!j_1...j_{n-1}k}\mathbb{J}^{k\,i_m}\Big(X^{m+n-2}\Big)^{\!i_1...i_{m-1}\,j_1...j_{n-1}}\!\right)
\end{equation*}
in which repeated tensorial indexes are summed. Having established the
notation for the various operations among moments, we turn next to
their application in geodesic Vlasov flows.

\subsection{Truncation of moment hierarchies}

The truncation of Vlasov moment hierarchies is a typical problem in kinetic
theory and, for statistical moments, this question was addressed by Channell
\cite{Ch1995} by using Levi's decomposition theorem. Channell's result may
be summarized by saying that, for a moment Hamiltonian not depending on  the first-order moment ($\pa H/\pa X^1=0$), the moment hierarchy can be truncated at any order, thereby yielding a truncated Lie-Poisson system. In order to see how this works in practice, let us write the truncated equations at the order $K$ for the hierarchy of statistical moments $\dot{X}^n={\rm ad}^*_{h_m}\,X^{m+n-2}$, where $h_m=\pa H/\pa X^m$. We have
\begin{align}\nonumber
\dot{X}^1&=\ad^*_{h_2}\,X^{1}+\ad^*_{h_3}\,X^{2}+\,\dots+\ad^*_{h_K}\,X^{K-1}
\\\nonumber
\dot{X}^2&=\ad^*_{h_2}\,X^{2}+\ad^*_{h_3}\,X^{3}+\,\dots+\ad^*_{h_K}\,X^{K}
\\
\dot{X}^3&=\ad^*_{h_2}\,X^{3}+\,\dots+\ad^*_{h_{K-1}}\,X^{K}
\\\nonumber
&\vdots
\\\nonumber
\dot{X}^{K-1}&=\ad^*_{h_2}\,X^{K-1}+\ad^*_{h_3}\,X^{K}
\\\nonumber
\dot{X}^{K}&=\ad^*_{h_2}\,X^{K}
\end{align}
and we recognize that the equation for $X^1$ is decoupled, so we restrict to
the equations for $X^2,\dots,X^K$. At this point, one verifies that the Lie
Poisson bracket for the truncated moment system is given by
\begin{equation}
\{F,G\}(X)\,=\!\sum_{n=2}^K\, \sum_{m=2}^{K-n+2}\left\langle
X^{n+m-2\,},\left[\frac{\partial F}{\partial X^n},\frac{\partial
G}{\partial X^m}\right]\right\rangle
\end{equation}
with the same notation as in (\ref{mombrkt}).
We recognize that the truncated structure is completely determined by the Lie algebra filtration (\ref{StatMomFiltr}) and does not depend
on the particular expression of the Lie bracket, which was not used in
deriving the truncated system above.

Following similar arguments, one can write the truncated system for kinetic moments as
\begin{align}\nonumber
\pa_t A_0&=-\ad^*_{h_1}\,A_{0}-\ad^*_{h_2}\,A_{1}-\,\dots-\ad^*_{h_K}\,A_{K-1}
\\\nonumber
\pa_t A_1&=-\ad^*_{h_1}\,A_{1}-\ad^*_{h_2}\,A_{2}-\,\dots-\ad^*_{h_K}\,A_{K}
\\
\pa_t A_2&=-\ad^*_{h_1}\,A_{2}-\,\dots-\ad^*_{h_{K-1}}\,A_{K}
\\\nonumber
&\vdots
\\\nonumber
\pa_t A_{K-1}&=-\ad^*_{h_1}\,A_{K-1}-\ad^*_{h_2}\,A_{K}
\\\nonumber
\pa_t A_{K}&=-\ad^*_{h_1}\,A_{K}
\end{align}
where $h_m=\delta H/\delta A_m$ and we have assumed $h_0\equiv0$. The zero-th moment equation decouples and one is left with the remaining equations for $A_1,\dots,A_K$. These equations possess the following bracket structure:
\begin{equation}
\{F,G\}(A)\,=-\,\sum_{n=1}^K\, \sum_{m=1}^{K-n+1}\left\langle
A_{n+m-1\,},\left[\frac{\delta F}{\delta A_n},\frac{\delta
G}{\delta A_m}\right]\right\rangle
\end{equation}
with the same notation as in (\ref{MomSchouten}). One recognizes again that the truncated structure is uniquely determined by the Lie algebra filtration (\ref{KinMomFiltr}). This fact suggests that a similar approach would also apply to the BBGKY moments of the Liouville equation (reduced probability distributions), whose corresponding Lie algebra is known to possess a similar filtration \cite{MaMoWe84}. However, we leave the investigation of the BBGKY moment hierarchy for another time.

\section{The geodesic Vlasov equation}
We have seen that moment hierarchies are equivalent descriptions of
the Vlasov equation, which allow for geometric closures of the
kinetic system (e.g. the ideal fluid closure for kinetic moments). The Vlasov equation is a Lie-Poisson equation on the Lie
algebra of the group $Can(T^*Q)$ of canonical transformations, and
this property is reflected in the Lie-Poisson structure of moment
dynamics. It is well known that physical systems with interesting geometric behavior are often geodesic flows on Lie
groups with respect to a metric provided by the system's  kinetic energy. The most familiar example is probably rigid body motion, which is governed by geodesic motion on SO(3). Likewise, Euler's equations for ideal fluids may be interpreted as geodesic motion on the volume-preserving diffeomorphisms $\rm Diff_{\rm vol}(\mathbb{R}^3)$ of the 3D flow domain $\mathbb{R}^3$  \cite{Ar1966}. Another
interesting example of geodesic motion is provided by the EPDiff equation \cite{HoMaRa1998}, which governs geodesic motion on the full diffeomorphism group Diff($\mathbb{R}^n$). In many cases, geodesic flows on Lie groups turn out to be completely integrable Hamiltonian systems. For example, in a one-dimensional flow domain $\mathbb{R}$, EPDiff recovers the Camassa-Holm equation for shallow water waves \cite{CaHo1993}.

The problem of geodesic flow on the symplectic group ({\it symplecto-hydrody-namics}) was  introduced by Arnold and Khesin in \cite{ArKe1998}. The present paper pursues this idea by considering geodesic flow on the canonical transformations within the context of Vlasov dynamics. In particular, we develop a {\bfi geodesic Vlasov equation} called EP$Can$ ({\it Euler-Poincar\'e equation on the canonical transformations}) as an extension of previous work in \cite{GiHoTr2005}.

\begin{remark}[Symplectomorphisms vs canonical transformations] $\quad$\\
The name  \textnormal{EP}$Can$ refers to the Euler-Poincar\'e
equation for geodesic motion on the subgroup of the symplectic
transformations arising from Hamiltonian vector fields. The subgroup
$Can$ (which could equally well be called \textnormal{Diff}$_{\rm can}$) may be
identified as the group of smooth invertible \emph{canonical}
transformations with smooth inverses. These transformations coincide
with symplectic transformations in simple domains such as
$\mathbb{R}^{2K}$. In those simpler cases, the geodesic Vlasov
equation is also known as ${\rm EP}$Symp \cite{GiHoTr2005,GiHoTr2007}
\end{remark}

As mentioned in the introduction, the idea to investigate  \textnormal{EP}$Can$ was motivated by
the observation that geodesic equations for kinetic moments were found to include
the Camassa-Holm equation \cite{GiHoTr2005}. In fact, geodesic moment equations arise from EP$Can$ whenever the norm may be expanded as a Taylor series. Later, it was recognized \cite{GiHoTr2007} that the geodesic moment equations also recover the {\it two-component Camassa-Holm equation} \cite{ChLiZh2005}. The latter is a geodesic flow on a semidirect-product Lie group \cite{Kuzmin2007}.  In order to explain these issues, we shall introduce the EPCan equation and show how it specializes to each integrable case.

Given a symplectic manifold $\mathcal{P}$ of even dimension $N=2K$, the EP$Can$ Vlasov Hamiltonian is defined by
\begin{equation}
H[f]=\frac12 \iint f(z)\, \mathcal{G}(z,z')\, f(z') \
\rd^{N\!}z\ \rd^{N\!}z'
=\frac12\,\big\|f\big\|_\text{\scriptsize$\mathcal{G}$}^{\,2}
\label{EPsympHam}
\end{equation}
where $z\in \mathcal{P}$ and the kernel $\mathcal{G}$ is chosen so that it
defines an appropriate norm on ${\rm Den}(\mathcal{P})$. When
dealing with moments, we shall restrict to the special cases
$\mathcal{P}=T^*Q$, with $Q$ a general configuration manifold, and
$\mathcal{P}=V$, with $V$ a symplectic vector space. The geodesic
Vlasov equation (aka EP$Can$) is written simply as
\begin{equation}\label{epcan}
\pp{f}{t}=-\left\{ f, \dede{H}{f}\right\}=-\Big\{ f,\,
\mathcal{G}*f\Big\}
\end{equation}
which coincides with Euler's vorticity equation in 2D when $\mathcal{G}=(-\Delta)^{-1}$.

\subsection{Euler-Poincar\'e equations on Hamiltonian vector fields}
In order to understand how the geodesic Vlasov equation arises from an Euler-Poincar\'e approach, one starts with an invariant Lagrangian defined on the tangent space of the canonical transformations
$\mathcal{L}:TCan\to\mathbb{R}$, which is purely quadratic. By the
invariance property, we can write the associated variational principle on
the Lie algebra of Hamiltonian vector fields as follows
\begin{equation}
\delta \int _{t_0}^{t_1}L[{\bf X}_h]\ {\rm d}t=0
\end{equation}
where ${\bf X}_h=\mathbb{J}\nabla h$ and
\begin{eqnarray}
L[{\bf X}_h]
&:=&
\frac12
\Big\langle\widehat{Q}{\bf X}_\textit{\small h},\,{\bf X}_\textit{\small h}\Big\rangle
=\frac12
\Big\langle\widehat{Q}\,\mathbb{J}\, \nabla h_{\,},\,\mathbb{J} \nabla h\Big\rangle
\nonumber\\
&=&
\frac12
\Big\langle{\rm div}\big(\mathbb{J}\,\widehat{Q}\,\mathbb{J}\, \nabla h\big),\, h\Big\rangle
=:L[h]
\,,
\end{eqnarray}
in which $\widehat{Q}$ is taken to be a positive-definite symmetric operator so that $L[{\bf X}_h]$ defines a nondegenerate norm.
The Legendre transform
\begin{equation}
f=\frac{\delta L}{\delta h}={\rm div}\big(\mathbb{J}\,\widehat{Q}\,\mathbb{J}\, \nabla h\big)
\,\Rightarrow\,
h=\left({\rm div}\,
\mathbb{J}\,\widehat{Q}\,\mathbb{J}\, \nabla
\right)^{\!-1}\!f
\end{equation}
yields the EP$Can$ Hamiltonian in the Vlasov form
\begin{equation}
H[f]=\frac12\,\big\langle f,\widehat{O}^{-1}f\big\rangle
\end{equation}
with
\begin{equation}
\widehat{O}:={\rm div}\,\mathbb{J}\,\widehat{Q}\,\mathbb{J}\, \nabla
\,.
\end{equation}
This formula specifies the relation between the geodesic Vlasov equation
and the geodesic motion on the Hamiltonian vector fields. An
interesting case occurs when $\widehat{Q}$ is the `flat' operation
$\widehat{Q}\,{\bf X}_\textit{\small h}={\bf X}_{\it h}^{\,\flat}$ which takes contravariant vectors to covariant vectors, so
that
\begin{equation}
{\rm div}\,
\mathbb{J}\left(\mathbb{J}\,\nabla h\right)^{\flat}=
-\,\Delta h
\,.
\end{equation}
Then the operator $\widehat{O}$ reduces to minus the Laplacian
\begin{equation}
\widehat{O}=-\,\Delta
\end{equation}
and in two dimensions one obtains the Euler Hamiltonian for vorticity dynamics,
\[
H[\omega]=1/2\,\big\langle\omega,(-\Delta)^{-1\,}\omega\big\rangle
\,,
\]
with $\omega=f$. In the more general case when $\widehat{Q}$ is a purely
differential operator, one finds that $\widehat{Q}$ and $\mathbb{J}$
commute and thus $\widehat{O}=-\,{\rm div}\,\widehat{Q}\, \nabla$. Also if
$\widehat{Q}$ commutes with the divergence, then, one has
$\widehat{O}=-\,\widehat{Q}\, \Delta$. However in general,
$\widehat{Q}$ is a matrix differential operator that does not commute
with $\mathbb{J}$.

\subsection{The ideal fluid dual pair for the Vlasov equation}
Both the Poisson-Vlasov system (\ref{VlasovPoisson}) in plasma
physics and the geodesic Vlasov equation (\ref{EPsympHam}) possess
the single-particle solution (Klimontovich solution) expressed as a
delta function in phase space,
\begin{equation}\label{klimsol}
f(\bz,t)=\sum_{a=1}^\mathcal{N}\, w_a\,\delta(\bz-\boldsymbol{\zeta}_a(t))
\,,
\end{equation}
where the index $a$ is summed over $a=1,\dots,\mathcal{N}$ and the number $w_a$ is a constant weight associated to each
particle. This apparently trivial solution is of fundamental
importance in physics and leads to the Klimontovich theory of
kinetic equations \cite{Kl1982}. In the remainder of this
paper, we shall see how such a solution emerges within the analysis
of integrable moment closures of EP$Can$.

In particular, an interesting situation occurs when one allows for
more general solutions of the form
\begin{equation}
f(\bz,t)=\sum_{a=1}^\mathcal{N}\int\!
w_a(s)\,\delta(\bz-\boldsymbol{\zeta}_a(s,t))\ {\rm d}^k s
\label{sing-fsoln-momap}
\end{equation}
where $s$ is a coordinate on the immersed submanifold
$S_a\subset\mathbb{R}^{2K}$ with $k={\rm dim}(S)$ (e.g. $k=1$ for a curve, or $k=2$ for a surface). For simplicity of notation, we have already suppressed the implied subscript $a$ in the arclength $s$ for each $w_a$ and $\boldsymbol{\zeta}_a$. For further simplicity and without loss of generality, we take $\mathcal{N} = 1$ and so suppress the index $a$ in what follows. This is equivalent to treating an isolated singular solution. As one might expect, this is only a notational simplification; not a real restriction.

The $\bzeta$'s in (\ref{sing-fsoln-momap}) belong to the space of embeddings
${\rm Emb}(S,\mathbb{R}^{2K})$. That is,
$\boldsymbol{\zeta}:S\hookrightarrow \mathbb{R}^{2K}$. Remarkably, these solutions define a momentum map%
\footnote{Here $\mathfrak{X}_{\rm can}$ denotes the Lie algebra of
Hamiltonian vector fields, which should not be confused with the
larger Lie algebra corresponding to the tangent space at the
identity of the symplectomorphism group $Symp$.}
\begin{equation}
{\bf J}_{\rm Sing}:\, {\rm
Emb}(S,\mathbb{R}^{2K})\to \mathfrak{X}_{\rm can}^*(\mathbb{R}^{2K})
\end{equation}
where $\mathfrak{X}_{\rm can}^*(\mathbb{R}^{2K})\simeq{\rm
Den}(\mathbb{R}^{2K})$. This momentum map
is produced by the {\it left} action of canonical transformations on
$\mathcal{P}$ by composition of functions; that is,
\[
\eta\cdot\bzeta=\eta\circ\bzeta
\,.
\]
Importantly, the same kind of momentum map arises in the motion of
ideal fluids  (e.g., for point vortices in 2D). See \cite{MaWe1983}, where these  momentum maps are shown to possess a dual pair structure \cite{We1983}. Namely, if one considers $S$ as a manifold with volume form
\[
\omega_{\rm vol}=w(s)\,{\rm d}^k s
\,,
\]
then the \emph{right} action of $\rm Diff_{\rm vol}$ on ${\rm
Emb}(S,\mathbb{R}^{2K})$
\[
\bzeta\cdot\eta=\bzeta\circ\eta
\,,
\]
yields another momentum map
\begin{equation}
{\bf J}_S:\,{\rm
Emb}(S,\mathbb{R}^{2K})\to\mathfrak{X}^*_{\rm vol}(S)
\,.
\end{equation}
In more generality, if $(S,w)$ is a volume manifold
and $(\mathcal{P},\omega)$ is a symplectic manifold, then the right
action momentum map is (cf. \cite{MaWe1983})
\begin{equation}
{\bf J}_S:\bzeta\mapsto\bzeta^*\omega:{\rm
Emb}(S,\mathcal{P})\to\mathfrak{X}^*_{\rm vol}(S)
\,.
\end{equation}
To summarize, we have the following dual pair structure\\
\begin{picture}(150,100)(-60,0)%
\put(110,75){${\rm Emb}(S,\mathcal{P})$}

\put(78,50){$\mathbf{J}_{\rm Sing}$}

\put(160,50){$\mathbf{J}_S$}

\put(72,15){$\mathfrak{X}^{\ast}_{\rm can} (\mathcal{P})$}

\put(170,15){$\mathfrak{X}^{\ast}_{\rm vol}(S)$}

\put(130,70){\vector(-1, -1){40}}

\put(135,70){\vector(1,-1){40}}

\end{picture}\\
which is formally equivalent to the dual pair structure for ideal
fluids \cite{MaWe1983}. Moreover, the left leg yields a solution of
the Vlasov equation regardless the number of dimensions and this
makes the above dual pair a natural object in kinetic theory.

In order to write explicit formulas, we specialize to the case
$\mathcal{P}=\mathbb{R}^{2K}$. Upon denoting $\bzeta(s)=({\bf
Q}(s),{\bf P}(s))$, one defines the following Poisson structure on
${\rm Emb}(S,\mathcal{P})$
\begin{equation}
\big\{F,G\big\}_{\rm Emb}=\sum_{i=1}^K\int \! \frac1{w(s)}\left(
\dede{F}{Q^i}\dede{G}{P_i}- \dede{G}{Q^i}\dede{F}{P_i}\right) {\rm
d}^k s
\end{equation}
where we see that the factor $1/w(s)$ is needed for functionals
of the form {$G({\bzeta})=\int\omega_{\rm vol}\,g(\bzeta)=\int
w(s)\,g(\bzeta(s))\, {\rm d}^k s$}, whose functional derivative
$\delta G/\delta \zeta=w(s)\,{\rm d}{g}/{\rm d}{\zeta}$ takes
values in ${\rm Den}(S)$.

\begin{remark}[Comparison with point vortices] $\quad$\\
The factor $1/w(s)$ is reminiscent of the vortex strength factors in
the Poisson bracket for point vortices \cite{MaWe1983}. Indeed, the
bracket above appears as the higher dimensional version of the
vortex bracket, so that the one dimensional vortex strengths are
replaced by appropriate densities (the weights $w(s)$) on the
embedded space $S$.
\end{remark}
Finally one checks that, for any Hamiltonian function $h\in
\mathcal{F}(T^*\mathbb{R}^{2K})$,
\begin{equation}
\big\{F,\langle{\bf J}_{\rm Sing},h\rangle\big\}_{\rm Emb}={\bf
X}_h[F]
\end{equation}
where ${\bf X}_h[F]$ is the infinitesimal generator of the action of
canonical transformations $Can$($\mathbb{R}^{2K}$) on ${\rm
Emb}(S,\mathbb{R}^{2K})$. Thus, ${\bf J}_{\rm Sing}$ satisfies the
classical definition of a momentum map.

The singular solution momentum map ${\bf J}_{\rm Sing}$ produces the collective Vlasov Hamiltonian
$H\circ{\bf J}_{\rm Sing}$. In particular, substituting the singular solution momentum map (\ref{sing-fsoln-momap}) into the EP$Can$ Hamiltonian (\ref{EPsympHam}) yields the collective Hamiltonian
\begin{equation}\label{EPCan-collective}
H_\mathcal{N}=\frac12\sum_{a,b=1}^\mathcal{N}\iint
w_a(s)\,w_b(s')\,\mathcal{G}\big({\bf Q}_a(s),{\bf P}_a(s),{\bf
Q}_b(s'),{\bf P}_b(s')\big)\ {\rm d}^ks\,{\rm d}^ks'
\end{equation}
thereby producing the following collective equations of motion
\begin{align*}
\frac{\partial {\bf Q}_a(s,t)}{\partial t} &= \frac{\delta
H_\mathcal{N}}{\delta {\bf P}_a} =
w_a(s)\sum_{b=1}^\mathcal{N}\int\! w_b(s')\,\frac{\partial
}{\partial {\bf P}_a}\,\mathcal{G}\big({\bf Q}_a(s),{\bf
P}_a(s),{\bf Q}_b(s'),{\bf P}_b(s')\big)\, {\rm d}^ks'
\\
\frac{\partial {\bf P}_a(s,t)}{\partial t} &=- \frac{\delta
H_\mathcal{N}}{\delta {\bf Q}_a} =-\,
w_a(s)\sum_{b=1}^\mathcal{N}\int\! w_b(s')\, \frac{\partial
}{\partial {\bf Q}_a}\,\mathcal{G}\big({\bf Q}_a(s),{\bf
P}_a(s),{\bf Q}_b(s'),{\bf P}_b(s')\big)\, {\rm d}^ks' .
\end{align*}

\begin{remark}[Possible divergent terms in collective motion]$\,$\\
The collective dynamics of singular solutions deserves some care,
depending on the form of the Vlasov Hamiltonian. The existence of
such solutions does not guarantee the existence of a well defined
collective Vlasov Hamiltonian $H\circ{\bf J}_{\rm Sing}$. This is
because  the singular solution momentum map ${\bf J}_{\rm Sing}$ may
produce divergent terms in the collective Hamiltonian \cite{We1983}.
For example, this is the case of the Vlasov Poisson system
(\ref{VlasovPoisson}), where the divergence is generated by
potential terms such as $1/2\sum w_a w_b|{\bf Q}_a-{\bf Q}_b|^{-1}$,
when $a=b$. The same situation occurs for point vortex solutions of
the planar Euler's vorticity equation; these solutions correspond to
the 2D phase space Hamiltonian $H_\mathcal{N}=1/2\,\sum
w_a\,w_b\,\log|(Q_a-Q_b,P_a-P_b)|$. On the other hand, these
problems are absent, for instance, in the Vlasov-Helmholtz system
(see \cite{GiHoTr2008} and references therein), since the potential
terms there are given by $1/2\sum w_a w_b \,e^{|{\bf Q}_a-{\bf
Q}_b|}$.
\end{remark}

As for the right-action momentum map, the expression
\begin{equation}
{\bf
J}_S(\bzeta)=\bzeta^*\omega=\bzeta^*({\rm d}{\bf q}\wedge{\rm
d}{\bf p})=\bzeta^*{\rm d}{\bf q}\wedge\bzeta^*{\rm d}{\bf p}={\rm d}\!\left(\bzeta^*{\bf
q}\right)\wedge{\rm d}\!\left(\bzeta^*{\bf p}\right)
\end{equation}
yields the
following simple expression
\begin{equation}\label{right-momap}
{\bf J}_S({\bf Q,P})={\rm d}{\bf Q}(s)\wedge {\rm
d}{\bf P}(s)=\sum_{n,m=1}^k\papa{\bf Q}{s^n}\papa{\bf P}{s^m}\,{\rm
d}s^n\!\wedge{\rm d}s^m
\end{equation}
The conservation law ${\rm d}{\bf J}_S/{\rm d}t=0$ is recovered by Noether's theorem, due to the
Diff$(S)$-invariance of the collective Vlasov Hamiltonian $H_\mathcal{N}=H\circ{\bf
J}_{\rm Sing}$ in (\ref{EPCan-collective}).

\rem{ 
\comment{CT: Darryl, I worked with Francois trying to prove the
circulation theorem below and we couldn't make it. It looks like
this works only for the LE map ($S=\mathcal{P}$). What do you
think?\\

DH: I believe you and Francois may be correct, Cesare. Does conservation of this momentum maps open the door for understanding the moment expressions for the Poincar\'e invariants found in \cite{HoLySc1990}?}
\begin{framed}
Thus, since this quantity is conserved (by the Diff$_{\rm
can}-$invariance of the collective Vlasov Hamiltonian $H\circ{\bf
J}_{\rm Sing}$), the associated Kelvin-Noether theorem is written as
\begin{equation}
\dd{}{t}\oiint_{\sigma(t)} \!{\rm d}{\bf Q}^{i}(s)\wedge {\rm d}{\bf
P}_{\!j}(s)\,=\,0
\end{equation}
where $\sigma(t)$ is a surface moving with the Hamiltonian flow.
This result reflects the conservation of the canonical symplectic
form.
\end{framed}
}    
When $S$ is a Lagrangian submanifold (this requires ${\rm
dim}(S)=1/2\,{\rm dim}(\mathcal{P})$), the momentum map ${\bf J}_S$
restricts to ${\bf J}_S(\bzeta)=0$. Likewise, the case ${\rm
dim}(S)=0$ recovers the usual Klimontovich solution (\ref{klimsol})
of particle motion used in kinetic theory. This fact, together with
the geometric results on moment hierarchies of kinetic equations,
illustrates the geometric basis of kinetic theory, in analogy to
Arnold's formulation of the ideal fluid \cite{Ar1966}.

\subsection{Klimontovich solution and the Lagrange-to-Euler
map} This section discusses the two limiting cases of the singular
solution momentum map ${\bf J}_{\rm Sing}$, that is ${\rm dim}(S)=0$
and $S=\mathcal{P}$. As mentioned above, the first case yields the
Klimontovich solution (\ref{klimsol}), which is then a momentum map
\makebox{${\bf J}_{\rm
Sing}:\text{\large$\times$}_{\!a=1\,}^\mathcal{N}\mathcal{P}\to\mathfrak{X}_{\rm
can}^*(\mathcal{P})$}. In this case, the solution identifies the
particle trajectories, subject to initial conditions
$\bzeta_a(0)=\bz_a^0$, so that the particles are transported in the
phase space $\mathcal{P}$ by canonical transformations as
$\left\{\bzeta_a(t)\right\}=\psi_t\circ\left\{\bz_a^0\right\}$,
where $\psi_t\in{Can}(\text{\large$\times$}_{\!a\,}\mathcal{P})$  is
generated by the collective Hamiltonian $H\circ{\bf J}_{\rm Sing}$.
In two phase-space dimensions, the Klimontovich solution is the
usual point vortex solution of the Euler's vorticity equation.

\rem{ 
\comment{CT: I find interesting that the action of ${Can}(\mathcal{P})$ on $\text{\large$\times$}_{\!i\,}\mathcal{P}$
produces transformations in ${Can}(\text{\large$\times$}_{\!i\,}\mathcal{P})$, where collective
Liouville dynamics takes place. Indeed, $H\circ{\bf J}_{\rm
Sing}:\text{\large$\times$}_{\!i\,}\mathcal{P}\to\mathbb{R}$ and The
Liouville distribution is defined in ${\rm
Den}\left(\text{\large$\times$}_{\!i\,}\mathcal{P}\right)\simeq\mathfrak{X}^*_{\rm
can}(\text{\large$\times$}_{\!i\,}\mathcal{P})$. I think this is the
power of the Klimontovich approach - is it obvious?}
}    

Another suggestive case of the above treatment is given by
$S=\mathcal{P}$, so we may denote $s=\bz^0$. Then, one has
$\bzeta_a(\cdot,t)=\eta^{(a)}_t\in{Can}(\mathcal{P})$
and the momentum map ${\bf J}_{\rm
Sing}:\text{\large$\times$}_{\!a\,}{Can}(\mathcal{P})\to\mathfrak{X}_{\rm can}^*(\mathcal{P})$ is
written as
\begin{align}
f(\bz,t)=\sum_a\int\!w_a(\bz^0)\,
\text{\large$\delta$}\!\left(\bz-\eta_t^{(a)\!\!}\cdot\bz^{0}\right)\,
{\rm d}^{2K} \bz^{0}
\,.
\rem{ 
\\
=\iint\! \text{\large$\delta$}\Big({\bf p}-{\bf P}\big({\bf
q}^0,{\bf p}^0,t\big)\Big)\, \text{\large$\delta$}\Big({\bf q}-{\bf
Q}\big({\bf q}^0,{\bf p}^0,t\big)\Big)\, {\rm d}^{K} {\bf q}^{0}
\,{\rm d}^{K} {\bf p}^{0}
}     
\end{align}
This expression coincides with the well known Lagrange-to-Euler map for fluids, whose
importance is well established in continuum dynamics. The
Lagrange-to-Euler map is equivalent to the characteristic form of
the Vlasov equation (\ref{VlasovChar}). Notice that both the
Klimontovich and the Lagrange-to-Euler maps are produced by the {\it
same} Lie group ${Can}(\mathcal{P})$ acting on
$\text{\large$\times$}_{\!a\,}\mathcal{P}$  and
$\text{\large$\times$}_{\!a\,}{Can}(\mathcal{P})$
respectively, with the same left action by composition of functions, that is
$\eta\cdot\left\{\bzeta_a\right\}=\left\{\eta\circ\bzeta_a\right\}$
in the first case and
$\eta\cdot\left\{\eta^{(a)}\right\}=\left\{\eta\circ\eta^{(a)}\right\}$
in the second. On the other hand, for the Klimontovich case, the
collective dynamics generated by the Hamiltonian
$H\circ{\bf J}_{\rm
Sing}:\text{\large$\times$}_{\!a\,}\mathcal{P}\to\mathbb{R}$
produces the canonical transformations $\psi\in{Can}(\text{\large$\times$}_{\!a\,}\mathcal{P})\neq\text{\large$\times$}_{\!a\,}{Can}(\mathcal{P})$. This point is of fundamental
importance because the Lie group ${Can}(\text{\large$\times$}_{\!a\,}\mathcal{P})$ is the symmetry
group of the Liouville equation \cite{MaMoWe84} and the difference
between  ${Can}(\text{\large$\times$}_{\!a\,}\mathcal{P})$ and $\text{\large$\times$}_{\!a\,}{Can}(\mathcal{P})$ is
related to the particle correlations, which are neglected in the
second situation. (The latter is the Vlasov mean field approximation.)

The fact that these two fundamental maps each arise from the left leg of a dual pair of momentum maps again illuminates the geometric footing of kinetic theory. The above arguments also provide
mathematical support for the wide success of Klimontovich method in
kinetic equations \cite{Kl1982}.

\subsection{Geometric kinetic theory}
The presence of dual pairs in kinetic theory illuminates the Liouville and Vlasov equations in the light of their Lie symmetry properties. (Something similar happens for Euler's vorticity equation.) Namely, the presence of
momentum maps is not accidental in kinetic approaches. Indeed,
a reasonable summary of the results in \cite{MaMoWe84,HoLySc1990,GiHoTr2008} could be made by saying that the process
\begin{equation}
\hspace{-.3cm} \fbox{Liouville equation}
\,\to\fbox{Vlasov equation}
\begin{matrix}
\,\
\nearrow\fbox{\,\ ideal fluid\,\ }\quad
\\\,\\
\!\!\!
\searrow\fbox{beam optics}
\end{matrix}
\end{equation}
is given by a composition of momentum maps. In other words, taking the moments (BBGKY, kinetic or statistical) of a
Lie-Poisson kinetic equation is always a momentum map \cite{MaMoWe84,HoLySc1990,GiHoTr2008}. Moreover, the
closures adopted to obtain Vlasov from BBGKY, fluid theory from
kinetic moments and beam optics from statistical moments are also momentum maps arising from particular subgroups of the symmetry group of the starting system. More explicitly, passing from Liouville to Vlasov
requires the subgroup ${Can}(\mathcal{P})\subset{Can}(\text{\large$\times$}_{\!i\,}\mathcal{P})$. Likewise, passing
from Vlasov to fluid requires the fiber preserving subgroup $Can_\pi(T^*Q)\subset{Can}(T^*Q)$. Finally, passing
from Vlasov to beam optics requires the subgroup ${\rm
Sp}(2K,\mathbb{R})\subset{Can}(\mathbb{R}^{2K})$. We
can summarize the situation in the following statement
\begin{framed}\noindent\it
All these moment approximations in kinetic theory are momentum maps. 
\end{framed}\noindent
That the BBGKY distributions are momentum maps is a remarkable fact. One may ask whether the Klimontovich averages in plasma theory also share this property. In this case, the autocorrelations considered in the latter approach would again be naturally included in the geometry of the theory. We leave this promising question open, as a direction for future research.

\rem{ 
\subsection{Cold-plasma singular solution}
Upon denoting $\bz=({\bf q,p})\in T^*Q$, the
following cold plasma approximation is a solution of EPCan
\begin{equation}
f({\bf q,p},t)=\rho({\bf q},t)\,\delta({\bf p-\Pi(q},t))
\end{equation}
or, by exchanging variables ${\bf q} \leftrightarrow {\bf p}$
\cite{GiHoTr2005,GiHoTr2007},
\begin{equation}
f({\bf q,p},t)=\psi({\bf p},t)\,\delta({\bf q-\Xi(p},t))
\end{equation}
which still keeps track of the single particle trajectory on the
base configuration manifold $Q$ (not on the whole phase-space
$T^*Q$). While the cold plasma solution is well known to be a
momentum map (plasma to fluid map) \cite{MaWeRaScSp1983}, the geometric nature (if
any) of the latter type of solution is still unknown.

\comment{Amplify, or drop?}
} 

\section{Moment closures of EP$Can$: integrable cases}

As explained in \cite{GiHoTr2005,GiHoTr2007}, the geodesic Vlasov
equation may be represented in terms of the moments. Indeed, upon
supposing that the metric $\mathcal{G}$ in (\ref{EPsympHam}) and
(\ref{epcan}) is sufficiently smooth, may can expand
$\mathcal{G}(\bz,\bz')$ in a Taylor series.

\subsection{Integrable closures of kinetic moments}
In this section, we present the kinetic moment hierarchy for EP$Can$.
Upon denoting $\bz=(\bf q,p)$, one may expand $\mathcal{G}$ in a Taylor series, as follows,
\begin{align*}
\mathcal{G}(\bz,\bz')\,=&\sum_{n,m=0}^\infty{\bf p}^n\otimes{\bf
p}^{\prime\,m\!} \contract \,G_{nm}({\bf q,q'})
\\
=&\ \sum_{n,m}\ \sum_{i_1...i_n}\sum_{j_1...j_m}\big({\bf
p}^n\big)_{i_1,...,i_n}\left({\bf
p'}^m\right)_{j_1,...,j_m}\big(G_{nm}({\bf
q,q'})\big)^{i_1,...,i_n,j_1,...,j_m}
\end{align*}
where $G_{nm}$ is now a contravariant tensor field of rank $n+m$.
Inserting this Taylor expansion in the EPCan Hamiltonian yields
(with the notation above)
\begin{equation}
H=\frac12\sum_{n,m=0}^\infty\int   A_n({\bf q}) \otimes A_m({\bf
q}')\,\contract\, G_{nm}({\bf q,q'})\ {\rm d}^K{\bf q}\, {\rm
d}^K{\bf q}' = \frac12\,\big\|\{A_n\}\big\|_G
\end{equation}
Thus, upon denoting
\begin{equation}
G_{nm}*A_m:=
\int G_{nm}({\bf q,q'})\,\contract \,
A_m({\bf q}')\ {\rm d}^K{\bf q}'
\end{equation}
the geodesic moment equations become
\begin{equation}
\pp{A_n}{t}=-\sum_{m,k=0}^\infty \text{\large
ad}^*_{\,G_{m\!k\,}*A_k}\,A_{n+m-1}
\end{equation}
where ad$^*$ is the coadjoint Lie-Schouten operator. The singular  solutions of the Vlasov moment hierarchy may be expressed in the following form:
\begin{equation}\label{KlimKinMom}
A_n({\bf q},t)=\int w(s)\  {\bf P}^n(s,t)\,\delta({\bf q-Q}(s,t))\ {\rm
d}^k s
\end{equation}
\rem{ 
and
\begin{equation}
A_n({\bf q},t)=\int {\bf p}^n\,\psi({\bf p},t)\,\delta({\bf
q-\Xi(p},t))\ {\rm d}^p {\bf p}
\end{equation}
} 
which is a momentum map ${\bf J}:{\rm Emb}(S,T^*Q)\to
\mathfrak{g}^*$, where $\mathfrak{g}$ is the kinetic moment algebra.

As we have seen in the preceding discussions, the moment algebra
possesses the important subalgebra $\mathfrak{g}_1=\mathfrak{X}(Q)$
of vector fields on the configuration manifold. In terms of
canonical transformations, this corresponds to Hamiltonian
generating functions that are linear in the momentum coordinate,
i.e. point transformations. These are cotangent lifts $T^*$Diff($Q$)
of diffeomorphisms on the configuration manifold $Q$
\cite{HoMa2004,GiHoTr2007}. Remarkably, when the moment hierarchy of
EPCan is closed such that $G_{11}=:G_1$ is the only non-vanishing
term of $G_{nm}$, we obtain the Hamiltonian on the one-form density
$A_1=:{\bf m}({\bf q})\cdot{\rm d}{\bf q}\otimes{\rm d}^K{\bf q}\in
\mathfrak{X}^*$
\begin{equation}
H=\frac12\iint {\bf m}({\bf q})\cdot G_1({\bf q,q'})\, {\bf m}({\bf
q}')\ {\rm d}^K{\bf q}\ {\rm d}^K{\bf q}'
\end{equation}
By using the property of the Schouten bracket
$[\beta_1,\,\alpha_n]=\pounds_{\beta_1}\,\alpha_n$ (where $\pounds$
denotes Lie derivative), one finds the EPDiff equation,
\begin{equation}
\pp{\bf m}{t}+\textit{\large\pounds}_{G_1*\bf m}\, {\bf m}=0
\,.
\end{equation}
EPDiff is the {\it Euler-Poincar\'e equation on the diffeomorphisms}
\cite{HoMa2004}. This equation has the important property of
exhibiting emergent singular $\delta$-like solutions from any confined smooth initial configuration. In 1D, the particular case
$G_1=(1-\alpha_2\partial^2)^{-1}$ yields the integrable {\it
Camassa-Holm equation}, which is well known in the community of
integrable systems.

Remarkably, if we also allow for $G_{00}=:G_0\neq0$, we obtain the
Hamiltonian
\begin{multline*}
H[A_0,A_1]=\frac12\iint A_1({\bf q})\cdot G_1({\bf q,q'})\, A_1({\bf
q}')\ {\rm d}^K{\bf q}\ {\rm d}^K{\bf q}' \\
+ \frac12\iint A_0({\bf q})\ G_0({\bf q,q'})\, A_0({\bf q}')\ {\rm
d}^K{\bf q}\ {\rm d}^K{\bf q}'
\end{multline*}
which yields a geodesic flow on the semidirect-product Lie group
${\rm Diff}\circledS\,\mathcal{F}$, introduced in
\cite{HoTrYo2007} in the context of image matching, discussed further in \cite{HoTr2009} and shown in numerical simulations to exhibit emergent singularities in both of its variables  \cite{HoOnTr2009}.
Interestingly enough, the EP(${\rm Diff}\circledS\,\mathcal{F}$)
equations (with notation $\beta_0\diamond
A_0=-\,\ad^*_{\beta_0}\,A_0$)
\begin{align*}
A_{0,t}+\textit{\large\pounds}_{G_1*A_1}\,A_0 &= 0
\,,\\
A_{1,t}+\textit{\large\pounds}_{G_1*A_1}\,{A_1} &= A_{0\,}
\text{\large$\diamond$}\ (G_0
* A_0) \,.
\end{align*}
are a geodesic flow on the {\it extended point transformations},
i.e. compositions of cotangent lifts and fiber translations \cite{MaWeRaScSp1983}.
\begin{remark}[Geodesic motion on fiber-preserving transformations]$\quad$\\
The semidirect-product Lie group ${\rm
Diff}_{\,}\circledS_{\,}\mathcal{F}$ is identified with the
compositions of cotangent lifts with fiber translations on the phase
space $T^*Q$ with coordinates $({\bf q,p})$. This identification  is
of fundamental importance in plasma physics \cite{MaWeRaScSp1983}.
It also yields the interpretation of $EP\!\left({\rm
Diff}_{\,}\circledS_{\,}\mathcal{F}\right)$ as a geodesic motion on
the Lie group ${Can}_{\pi}(T^*Q)$ of {\bfi fiber-preserving
canonical transformations} on the cotangent bundle $T^*Q$. In fact,
any transformation by  ${Can}_{\pi}(T^*Q)$ can be realized as the
composition of a fiber translation and a cotangent lift (or
viceversa) \cite{BaWe1997}. On the other hand, such a transformation
is always a canonical transformation characterized by a generating
function that is \emph{linear} (and inhomogeneous) in the canonical
momentum. Therefore, since ${\rm
Diff}(Q)_{\,}\circledS_{\,}\mathcal{F}(Q)\simeq {Can}_{\pi}(T^*Q)$
then $EP({\rm Diff}\circledS_{\,}\mathcal{F})\simeq EP{Can}_{\pi}$.
\end{remark}

In 1D, the special case $G_0=\delta$-function yields yet another integrable system,
\begin{align*}
\lambda_t&=-\left(u\lambda\right)_q
\\
u_t-u_{qqt}&=-3uu_q+2u_q u_{qq}+u u_{qqq} -\lambda\lambda_q
\end{align*}
known as {\it two-component Camassa-Holm equation} \cite{ChLiZh2005}
(here $u=(1-\partial^2)^{-1}A_1$ and $\lambda=A_0$). This system
first appeared in \cite{OlRo1996}. Its singular solutions were
studied in \cite{CoIv2008} who pointed out their relation between this
system and shallow water equations. Upon slightly modifying the
Hamiltonian by $G_0=(1-\partial^2)^{-1}$, one has the spontaneous
emergence of the Klimontovich solution (\ref{KlimKinMom})  in $(A_0,A_1)$, as shown in \cite{HoOnTr2009}. The two-component Camassa-Holm system and its Klimontovich solutions have
also found applications in the \emph{metamorphosis} approach to  image matching, e.g., for magnetic resonance images \cite{HoTrYo2007}.

\subsection{Geodesic flow on statistical moments}

Section 5 shows that representing the geodesic Vlasov equation in
terms of statistical moments in a special case recovers the
well-known Bloch-Iserles integrable system
\cite{BlIsMaRa2005,BlIs2006}. If $\mathcal{G}$ in (\ref{EPsympHam})
and (\ref{epcan}) is assumed to be analytic in {\it both} position
and momentum, its Taylor expansion may be written as:
\begin{equation}
\mathcal{G}(\bz,\bz')=\sum_{n,m=0}^\infty\Gamma_{nm}\contract\,\bz^{n}\otimes\bz^{\prime\,m\!\!}
\end{equation}
where $\Gamma_{nm}$ is a covariant tensor of rank $n+m$. Then, the
EP$Can$ Hamiltonian (\ref{EPsympHam}) is written in terms of the moments
$X^n(t)=\int{\bf z}^n f({\bf z},t)\ {\rm d}^N{\bf z}$ as
\begin{equation}
H=\frac12 \sum_{n,m=0}^\infty \Gamma_{nm}\contract\,  X^n\otimes X^m
\end{equation}
so that the moment equations for $(n,m)\in \mathbb{Z}_+$ turn out to be
\begin{equation}
\dot{X}^{m}\,= -\,m\sum_{n=0}^\infty n\
\text{\large$\mathcal{S}$}\bigg( \Big(\left(\Gamma_{nk}\contract\,
X^k\right)\cdot\mathbb{J}\Big)\contract \,X^{m+n-2}\bigg)
\label{mom-eqn-nm}
\end{equation}
where, as before, the symbol $\contract$ denotes contraction between
upper and lower indices of the various tensors. In explicit index
notation this reads as
\begin{multline*}
\Big(\dot{X}^m\Big)^{\!i_1...i_m}=
\\
-\,m\sum_{n=0}^\infty n\ \text{\large$\mathcal{S}$}\!\left(
\big(\Gamma_{n k}\big)_{\!j_1...j_n\,l_1...l_k}
\left(X^{k}\right)^{l_1...l_k}
\mathbb{J}^{j_n\,i_m}\Big(X^{m+n-2}\Big)^{\!i_1...i_{m-1}\,j_1...j_{n-1}}\!\right)
\,.
\end{multline*}
An example of this equation is discussed in Section
\ref{BlIs-eqn-sec} for the case $m=n=2$.

\subsection{Klimontovich dynamics for statistical moments}
As in the case of kinetic moments, the single-particle Klimontovich  solution
\begin{equation}
f(\bz,t)=\sum_a w_a\,\delta(\bz-\bzeta_a(t))
\label{Klim-soln}
\end{equation}
of the geodesic Vlasov equation offers an interesting opportunity
for presenting solutions of the moment hierarchy.
\rem{ 
Since
statistical moments are finite dimensional tensors (rather than
tensor fields), we have dropped the dependence on the initial
conditions in the Klimontovich Lagrange-to-Euler map thereby
simplifying the notation.
}  
These solutions may be represented as
\begin{equation}
X^n(t)=\sum_a w_a\,\bzeta_a^n(t)
\end{equation}
where $\bzeta_a^n:=\bzeta_a\otimes\dots\otimes\bzeta_a$ ($n$ times) and
$\bzeta_a$ satisfies Hamilton's canonical equations with Hamiltonian
\begin{equation}
H_\mathcal{N}
=
\frac12
\sum_{a,b}\sum_{n,m} w_a\,w_b\,\Gamma_{nm}\contract\ \bzeta_a^n\otimes\bzeta_b^m
\,.
\end{equation}
This means that $\bzeta_a$ satisfies
\begin{equation}
\dot{\bzeta}_a=\mathbb{J}\,\nabla_{\!\bzeta_a}H_{\!\mathcal{N}}
=
\mathbb{J}\sum_{nm}n\,w_a\,\sum_{b}w_b\,\Gamma_{nm}\contract\ \bzeta_a^{n-1}\otimes\bzeta_b^m
\,,
\end{equation}
where the moment order $m$ or $n$ ranges from zero to infinity, and $a,b$ range over the number of particles $a,b=1,\dots,\mathcal{N}$.
Each $\bzeta_a$ undergoes Hamiltonian dynamics because the
 Vlasov single-particle solution is an equivariant momentum map.
Equivariance of the corresponding momentum map for dynamics of point
vortices in an Euler fluid was proven in \cite{MaWe1983}. The same
argument holds for the singular solutions of kinetic moment
equations, including the EPDiff equation \cite{HoMa2004}.

\begin{remark}[Truncation of geodesic moment hierarchies]$\,$\\
It is important to notice that higher-order moment truncations do not possess the Klimontovich solution. The latter exists only for genuine moment \emph{closures}, such as the fluid closure for the kinetic moments (consisting of $A_0,A_1$), or the $2^{nd}$-order closure for statistical moments used in linear beam optics (consisting of $X^1,X^2$).
\end{remark}

\section{Bloch-Iserles system as a moment equation}
\label{BlIs-eqn-sec}

The case $(n,m)=(2,2)$ of the moment equation (\ref{mom-eqn-nm}) yields an important moment subalgebra, given by
{\bfi homogeneous quadratic polynomials on phase space}, i.e.
quadratic forms on $V$. In this case, formula (\ref{mombrkt}) implies the following Lie-Poisson moment bracket
\begin{equation}
\{F,G\}(X^2)=4 \left\langle X^{2},\,
\mathcal{S} \left(\frac{\partial F}{\partial
X^2}\ \mathbb{J}\,\frac{\partial G}{\partial
X^2}\right)\right\rangle
\end{equation}
where the moment $X^2$ is now a $N\times N$ symmetric  matrix.
Because of the antisymmetry of $\mathbb{J}$, the bracket above may
be rewritten as
\begin{equation}
\{F,G\}(X^2)=\left\langle X^{2\,},\left[\frac{\partial F}{\partial
X^2},\frac{\partial G}{\partial
X^2}\right]_{2\mathbb{J}}\right\rangle
\end{equation}
where
\begin{equation}
\left[\frac{\partial F}{\partial X^2},\frac{\partial G}{\partial
X^2}\right]_{2\mathbb{J}\!} := \frac{\partial F}{\partial
X^2}\,2\mathbb{J}\,\frac{\partial G}{\partial X^2} - \frac{\partial
G}{\partial X^2}\,2\mathbb{J}\,\frac{\partial F}{\partial X^2}
\end{equation}
which is the Lie bracket for the integrable Bloch-Iserles (BI)
system of equations introduced in \cite{BlIsMaRa2005,BlIs2006}. In
that case, an antisymmetric matrix $\mathbb{N}$ of any dimension
defines the following Lie bracket on the space of symmetric matrices
of the same dimension. In formulas one has
\begin{equation}\label{BILB}
[X,Y]_\mathbb{N}:=X\mathbb{N}Y-Y\mathbb{N}X
\end{equation}
where $X$ and $Y$ are symmetric matrices. The Bloch-Iserles system
\begin{equation}\label{BIeq}
\dot{X}=\left[(X)^2,\mathbb{N}\right]
\end{equation}
is Lie-Poisson on this Lie algebra, with Hamiltonian
$H_{BI}=\frac12{\rm Tr}(^{t\!}XX)$. Here the notation $(X)^2$
denotes standard matrix multiplication of $X$ by itself, in order to
distinguish from second-order moments $X^2$. In addition, $^{t\!}X$
denotes the transpose of the matrix $X$, so that $^{t\!}X=X$ when
$X$ is symmetric. One concludes the following.
\begin{theorem}
In the even-dimensional case, the integrable Bloch-Iserles system is the Vlasov moment equation (\ref{Vlasov-moment-eqn}) obtained from the quadratic Hamiltonian
\begin{equation}
H_{BI}=\frac12\langle X^2,X^2\rangle
\end{equation}
associated with the antisymmetric matrix $\mathbb{N}:=2\mathbb{J}$.
\end{theorem}

Let us now look at the moment bracket for functions of
\begin{equation}
(X^0,X^1,X^2)\in
\mathfrak{g}_0^*\oplus\mathfrak{g}_1^*\oplus\mathfrak{g}_2^* \simeq
\mathbb{R}\oplus V \!\oplus \left(V\vee V\right)
\end{equation}
where $\vee$ is again the symmetric tensor product defined in
(\ref{V-devil-def}) and $X^0=const$ is naturally taken as the
probability normalization. The moment bracket (\ref{mombrkt})
becomes
\begin{multline}\label{mombrkt-symtens}
\{F,G\}(X)
=\, X^0 \, \frac{\partial F}{\partial
\bX^1}\cdot\mathbb{J}\frac{\partial G}{\partial \bX^1}
\\
+
 \bX^{1}\cdot\left(\frac{\partial F}{\partial X^2}\,2\mathbb{J}\frac{\partial G}{\partial \bX^1}
-\,\frac{\partial G}{\partial X^2}\,2\mathbb{J}\frac{\partial
F}{\partial \bX^1}\right)
\\
+
 \left\langle X^{2\,},\left[\frac{\partial F}{\partial X^2},\frac{\partial G}{\partial X^2}\right]_{2\mathbb{J}}\right\rangle
\end{multline}
which is given by the direct sum of the canonical Poisson bracket on
$V$ in first term, plus the semidirect-product Lie-Poisson bracket on
$\mathfrak{g}_{2}^{*\,}\circledS\ \mathfrak{g}_{1}^*\simeq\,{\rm
Sym}\,\circledS\ V$ in the second and third terms. Thus, the specialization of the moment bracket (\ref{mombrkt}) to (\ref{mombrkt-symtens}) in this case defines a Lie-Poisson bracket on $\left({\rm
Sym}\,\circledS \,V\right)\oplus \mathbb{R}$.

We now turn our attention to the odd-dimensional BI system.  In this
system, one has a degenerate antisymmetric matrix $\mathbb{N}$ of
odd dimension $n$ and rank $2K$. Upon defining $\bar{\mathbb{N}}$ as
the non-degenerate minor of maximal dimension ($2K$), the degenerate
matrix $\mathbb{N}$
\begin{equation*}
\mathbb{N} = \begin{bmatrix}
\bar{\mathbb{N}} & 0  \\
0 & 0
 \end{bmatrix}
\end{equation*}
produces the Lie bracket (\ref{BILB}) associated to the BI equation
(\ref{BIeq}).

The odd-dimensional BI system is known to be a
geodesic flow on the space \makebox{$\left({\rm Sym}(2K)\,\circledS
\,\mathcal{M}_{2K\times d}\right)\oplus {\rm Sym}(d)$} endowed with
the Lie bracket (cf. equation (2.14) in \cite{BlIsMaRa2005})
\begin{multline} \label{ExtendedSDBracket}
\left[(S, A, B), (S', A', B')\right] : =
\\
 \left(S\,\bar{\mathbb{N}} S' - S'\,
\bar{\mathbb{N}} S\,,\, S \,\bar{\mathbb{N}} A' - S'\,
\bar{\mathbb{N}} A\,,\, ^{t\!}A \,\bar{\mathbb{N}} A' -\,^{t\!}A'\,
\bar{\mathbb{N}} A\right)
\end{multline}
where one denotes $d=n-2K$, for any $S, S' \in \operatorname{Sym}
(2K)$, $A, A' \in \mathcal{M}_{2K\times d}$, and $B, B' \in
\operatorname{Sym}(d)$.

We will show that the bracket (\ref{mombrkt}) for the Vlasov moment
system is Lie-Poisson on the dual to the Lie algebra
\makebox{$\left({\rm Sym}(2K)\,\circledS \,V\right)\oplus
\mathbb{R}$}. That is, we choose $d=1$, which is the case of
interest here. Proposition 2.5 in \cite{BlIsMaRa2005} shows that the
geodesic equations on ${\rm Sym}(2K+1)$ are equivalent to the
geodesic equations on \makebox{$\left({\rm Sym}(2K)\,\circledS
\,V\right)\oplus \mathbb{R}$}, where $V$ is a $2K$-dimensional
symplectic space carrying a non-degenerate symplectic structure
$\mathbb{\bar{N}}/2$. This leads us to the identification of the
Bloch-Iserles system with geodesic moment dynamics.
\begin{theorem}
When $d=1$, the odd-dimensional BI system is a Vlasov moment
equation of the form (\ref{Vlasov-moment-eqn}) on the statistical
moments. This equation is generated by the quadratic Hamiltonian
\begin{equation}
H(X)=\frac12 \left\langle
X^{2\,},X^2\right\rangle+\frac12\,\bX^1\cdot\bX^1
\end{equation}
yielding a geodesic moment flow. The corresponding BI Hamiltonian
$H_{BI}(X)=1/2\,{\rm Tr}(^{t\!}X X)$ is written in terms of the
symmetric $(2K+1)$-dimensional matrices of the form
\begin{equation*}
X = \begin{bmatrix}
X^2 & \bX^1  \\
^{t\!}\bX^{1} & 4X^0
\end{bmatrix}
\end{equation*}
These matrices are endowed with the Bloch-Iserles Lie bracket
(\ref{BILB}), where the (degenerate) antisymmetric matrix
$\mathbb{N}$ takes the form
\begin{equation*}
\mathbb{N} = \begin{bmatrix}
2\mathbb{J} & 0  \\
0 & 0
\end{bmatrix}
\end{equation*}
and $\mathbb{J}$ is the canonical symplectic matrix.
\end{theorem}
\begin{proof}
Upon denoting
\begin{equation}
\bar{\mathbb{N}}=2\mathbb{J}\,,\qquad S=\frac{\partial F}{\partial
X^2}\,, \qquad S'= \frac{\partial G}{\partial X^2}\,, \qquad
A=\frac{\partial F}{\partial \bX^1}\,, \qquad A'=\frac{\partial
G}{\partial \bX^1}\,,
\end{equation}
one sees that the only difference between the Lie bracket in
(\ref{ExtendedSDBracket}) and the bracket in (\ref{mombrkt-symtens})
resides in a constant factor in the first term of
(\ref{mombrkt-symtens})
\begin{equation}
 \left(\frac{\partial F}{\partial
\bX^1}\right)\cdot\mathbb{J}\frac{\partial G}{\partial \bX^1} =
\frac14\left[\left(\frac{\partial F}{\partial
\bX^1}\right)\cdot2\mathbb{J}\frac{\partial G}{\partial \bX^1} -
\left(\frac{\partial G}{\partial
\bX^1}\right)\cdot2\mathbb{J}\frac{\partial F}{\partial \bX^1}
\right]
\end{equation}
where the square bracket in the right hand side is identical to the
last component of (\ref{ExtendedSDBracket}). This difference however
can be easily overcome. Indeed, one can always re-define the Lie
bracket on \makebox{$\left({\rm Sym}(2K)\,\circledS \,V\right)\oplus
\mathbb{R}$} as
\begin{multline*} 
\left[(S, A, B), (S', A', B')\right] : =
\\
 \left(S\,\bar{\mathbb{N}} S' - S'\,
\bar{\mathbb{N}} S\,,\, S \,\bar{\mathbb{N}} A' - S'\,
\bar{\mathbb{N}} A\,,\,\frac14\left( ^{t\!}A \,\bar{\mathbb{N}} A' -
\,^{t\!}A'\, \bar{\mathbb{N}} A\right)\right)
\end{multline*}
and verify that the map
\begin{align*}
\Psi:\makebox{$\left({\rm Sym}(2K)\,\circledS \,V\right)\oplus
\mathbb{R}$}&\ \to {\rm Sym}(2K+1)
\\
(S,A,B)&\ \mapsto
\begin{bmatrix}
S & A  \\
^{t\!} A & 4B
 \end{bmatrix}
\end{align*}
is a Lie algebra isomorphism for ${\rm Sym}(2K+1)$, which is endowed
with the BI Lie bracket
$[X,Y]_{\mathbb{N}}=X\mathbb{N}Y-Y\mathbb{N}X$. The isomorphism
property is a direct verification identical to Proposition 2.5 in
\cite{BlIsMaRa2005}. In particular, let $(S, A, B), (S', A', B') \in
\left({\rm Sym}(2K)\,\circledS \,V\right)\oplus \mathbb{R}$ and
compute directly that
\begin{align*}
\Psi\big([(S, A, B),\, &(S', A', B')]\big) =
\\
&= \Psi\big(S\,\bar{\mathbb{N}} S' - S' \,\bar{\mathbb{N}} S,\, S
\,\bar{\mathbb{N}} A' - S' \,\bar{\mathbb{N}} A,\, 1/4\, (\,^{t\!}A
\,\bar{\mathbb{N}} A' -\, ^{t\!}A'
\,\bar{\mathbb{N}} A)\big) \\
& 
=
\begin{bmatrix}
S\,\bar{\mathbb{N}} S' - S' \,\bar{\mathbb{N}} S & S \,\bar{\mathbb{N}} A' - S' \,\bar{\mathbb{N}} A  \\
^{t\!}(S \,\bar{\mathbb{N}} A' - S' \,\bar{\mathbb{N}} A) & ^{t\!}A
\,\bar{\mathbb{N}} A' -\, ^{t\!}A' \,\bar{\mathbb{N}} A
 \end{bmatrix} \\
& 
=
\begin{bmatrix}
S & A  \\
^{t\!}A & 4B
 \end{bmatrix}
\begin{bmatrix}
\bar{\mathbb{N}} & 0  \\
0 & 0
 \end{bmatrix}
 \begin{bmatrix}
S' & A'  \\
^{t\!}A' & 4B'
 \end{bmatrix}
 -  \begin{bmatrix}
S' & A'  \\
^{t\!}A' & 4B'
 \end{bmatrix}
 \begin{bmatrix}
\bar{\mathbb{N}} & 0  \\
0 & 0
 \end{bmatrix}
 \begin{bmatrix}
S & A  \\
^{t\!}A & 4B
 \end{bmatrix} \\
 & 
 = \big[\Psi(S,A,B), \Psi(S',A',B')\big]_{\!\mathbb{N}}
\end{align*}
as required.
\end{proof}

Thus, we conclude that the BI system and the geodesic moment
equations are equivalent.

\subsection{Klimontovich solutions of the Bloch-Iserles system}
\label{Ksoln-BI-subsec}

The Klimontovich map for the geodesic Vlasov equation provides
simple solutions of the Bloch-Iserles system in any dimension. For
example, in even dimensions, the dynamics of the BI solution
\begin{equation}\label{BIKlim-even}
X(t)=\sum_a w_a\,\bzeta_a^2(t)
\end{equation}
is given by the system
\begin{equation}
\dot{\bzeta}_a= w_a\sum_{b}w_b\, \mathbb{N}\,\bzeta_b^2\, \bzeta_a
\end{equation}
where $\mathbb{N}=2\,\mathbb{J}$ and $\bzeta_b^2\, \bzeta_a$ is a
covector such that $(\bzeta_b^2\, \bzeta_a)_i=(\bzeta_b^2)_{ij}\,
(\bzeta_a)^j$. In explicit index notation, one has
\[
\big(\dot{\bzeta}_a\big)^i= 2\,w_a\sum_{b}w_b\,
\mathbb{J}^{ij}\left(\bzeta_b^2\right)_{jk} \left(\bzeta_a\right)^k
\,.
\]
This system is a Hamiltonian system with the homogeneous quartic
Hamiltonian
\begin{equation}
H_\mathcal{N}=\frac12\sum_{a,b=1}^\mathcal{N}w_a\,w_b\,{\rm Tr}\left(
\,^{t\!}(\bzeta_a^2)\,\bzeta_b^2\right)
\end{equation}
Notice that, by writing the equation for $\bzeta_a$ as
\begin{equation}
\dot{\bzeta}_a= w_a\sum_{b\neq a}w_b\, \mathbb{N}\,\bzeta_b^2\,
\bzeta_a+w_a^2\,\|\bzeta_a\|^2\,\mathbb{N}\, \bzeta_a
\end{equation}
we can specialize the above to the simple case when $w_a=1$ for a
fixed $a$ and $w_b=0\ \forall b\neq a$, so that
\begin{equation}
X(t)=\bzeta^2(t) \qquad\text{ with }\qquad \dot{\bzeta}=
\left\|\bzeta\right\|^2 \mathbb{N} \,\bzeta
\end{equation}
This case leads however to trivially linear dynamics, since the norm
$\left\|\bzeta\right\|$ is evidently conserved. This does not happen
for different norms in the quadratic moment Hamiltonian, such as
${H}=1/2\left(\Gamma_{22}\contract\, X^2\otimes X^2\right)$.

The above arguments also provide solutions to the Bloch-Iserles
system in $2K+1$ dimensions. Indeed, the particle solution of
EP$Can$ (\ref{Klim-soln}) becomes a solution of the BI system in the
following form
\begin{equation}\label{BIKlim-odd}
X(t) = \sum_a w_a\! \begin{bmatrix}
\ \boldsymbol{\zeta}_a^2(t) & \boldsymbol{\zeta}_a(t)  \ \\
^{t\!}\boldsymbol{\zeta}_a(t) & 4  \
 \end{bmatrix}
 \in{\rm Sym}(2K+1)
 \,.
\end{equation}
Here $\bzeta_a$ undergoes Hamiltonian dynamics with
\begin{equation}
H_\mathcal{N}=\frac12\sum_{a,b}w_a\,w_b\,\Big(\bzeta_a\cdot\bzeta_b+
{\rm Tr}\left( \,^{t\!}(\bzeta_a^2)\,\bzeta_b^2\right)\Big)
\,,
\end{equation}
whose collective EP$Can$ equations are
\begin{equation}
\dot{\bzeta}_a= w_a\sum_{b}w_b\, \mathbb{N}\,\bzeta_b^2\,
\bzeta_a+\frac12\,w_a\sum_{b}w_b\,\mathbb{N}\,\bzeta_b
\,.
\end{equation}
As we shall see, these solutions are also momentum maps in both the
even and odd-dimensional cases, since the Klimontovich solution
(\ref{Klim-soln}) is a momentum map. In particular, upon fixing
$\mathcal{N}=1$, these are solution momentum maps $V\to{\rm Sym}(n)$
($\bzeta\mapsto X$), where $V$ is a symplectic space and $X\in{\rm
Sym}(n)$ is the Bloch-Iserles dynamical variable. This construction
arises from a special case of the momentum map in
(\ref{sing-fsoln-momap}), with dim($S$)=0 (Klimontovich case),
$\mathcal{P}=V$ (symplectic vector space) and where the Lie algebra
$\mathfrak{X}_{\rm can}(V)$ is restricted to the Lie subalgebra of
linear Hamiltonian vector fields. In the more general case when
dim$(S)\geq1$ one has also the conserved quantity in
(\ref{right-momap}). That is, the operation of taking moments
preserves the dual pair structure of the Vlasov equation.
\begin{theorem}
Upon fixing $\mathcal{N}=1$, the solution (\ref{BIKlim-odd}) of the
odd-dimensional Bloch-Iserles system is a momentum map
\[
{\bf J}_{2K+1}:(V,w\mathbb{J})\to{\rm Sym}(2K+1)
\]
where $w\mathbb{J}$ is the symplectic form on the vector space $V$.
Moreover, the solution (\ref{BIKlim-even}) in the even-dimensional
case is also a momentum map
\[
{\bf J}_{2K}:(V,w\mathbb{J})\to{\rm Sym}(2K)
\,.
\]
\end{theorem}
\begin{proof} In what follows we shall use the isomorphisms
\[
{\rm Sym}(2K+1)\simeq \left({\rm Sym}(2K)\,\circledS\,
V\right)\oplus\mathbb{R} \] and
\[
{\rm Sym}(2K)\simeq\mathfrak{sp}^*(2K,\mathbb{R}) \,,
\]
where $\mathfrak{sp}(2K,\mathbb{R})$ denotes the Lie algebra of
Hamiltonian matrices. We prove the first statement, which comprises
the second as a particular case. Let $V$ be endowed with the Poisson
structure
\[
\left\{F,G\right\}=\frac1w\  \Big.^{t}\!\!\left(\frac{\partial
F}{\partial \bzeta}\right)\cdot\mathbb{J} \frac{\partial G}{\partial
\bzeta}
\]
then, the definition of momentum map can be verified by inserting
$G=\left\langle{\bf J},\beta\right\rangle$, with ${\bf
J}=w(\bzeta^2,\bzeta,1)\in\left({\rm Sym}(2K)\,\circledS\,
V\right)\oplus\mathbb{R}$ and
$\beta=(\beta_2,\boldsymbol{\beta}_1,\beta_0)$ its dual. Then, we
have
\[
\left\{F,\left\langle{\bf J},\beta\right\rangle\right\}=
\Big.^{t}\!\Big(\mathbb{J}\cdot\boldsymbol{\beta}_1+2\mathbb{J}\beta_{2\,}\bzeta\Big)
\cdot\frac{\partial F}{\partial \bzeta}
\]
which identifies the infinitesimal action of linear (inhomogeneous)
Hamiltonian vector fields on the phase space functions
$F\in\mathcal{F}(V)$.

Restricting to even $2K$ dimensions requires setting
$\beta_0=0=\boldsymbol{\beta}_1$, thereby producing the action of
(homogeneous) Hamiltonian vector fields, i.e. Hamiltonian matrices
in $\mathfrak{sp}(2K,\mathbb{R})$.
\end{proof}

The integrability properties of these solutions will be discussed
elsewhere.

\section{Conclusions and outlook}

After reviewing the geometric basis of Vlasov moment dynamics, this
paper showed how moment closures of the geodesic Vlasov equation
(\ref{epcan}) produce interesting known integrable systems, which
include CH, CH2 and Bloch-Iserles (BI) equations in both odd and
even dimensions. While the CH and CH2 cases were already known to
arise in 1D \cite{GiHoTr2005,GiHoTr2007}, the higher dimensional
moment bracket showed that these moment closures also recover the
EPDiff equation and extend CH2 to higher dimensions. The paper also
recovered the BI system from (\ref{epcan}) by a finite-dimensional
moment closure, corresponding to inhomogeneous quadratic phase-space
functions. Thus, a kinetic theory approach led to special solutions
of BI.

The moment closures preserve the two equivariant momentum maps in
the Vlasov dual pair. Preservation of this structure guaranteed that
the resulting closed moment systems discussed here were still
Poisson. This preservation also enabled reduction to
finite-dimensional systems by using the Klimontovich particle
solutions from plasma theory. For example, the peakon solutions of
the CH equation arose from a Klimontovich approach in
\cite{GiHoTr2005}. Singular solutions also arose upon allowing extra
smoothing in the CH2 moment Hamiltonian
\cite{GiHoTr2007,HoOnTr2009}. In addition, this paper showed in
Section \ref{Ksoln-BI-subsec} that the same approach also produced
solutions of the finite-dimensional BI system. However, Klimontovich
solutions are not admitted by arbitrary approximations. They are
prevented, for example, when  moment hierarchies are simply
truncated at a certain weight. That is, moment {\it closures}
preserve the Vlasov dual pair, while moment {\it truncations} do
not, even though they may be shown to still be Lie-Poisson. Open
questions concern both the construction of a Lax pair for the
Klimontovich dynamics of the BI system and potential integrability
properties of the truncated equations (e.g. the ($A_1,A_2$)
truncation for kinetic moments).

The geometric setting showed how the left momentum map in the
Vlasov dual pair recovers both the Klimontovich solution and the
Lagrange-to-Euler map. That this geometry applies also to
the Liouville equation illuminates the geometric footing of
kinetic theory. Indeed, this paper explained how all the standard moment approximations in kinetic theory are momentum maps preserving the same dual pair. This construction would certainly be destroyed by introducing the collision integral, whose celebrated Boltzmann version implies irreversibility and produces a preferred direction of time via the $H$-theorem. This irreversibility prevents the Klimontovich solutions, which are solutions of a time-reversal invariant system.

The kinetic-theory interpretation of geodesic Vlasov moment dynamics
also provided insight into the physical description of the
integrable cases. For example, the CH2 case was interpreted in this
light as a charged fluid in the context of the one-component CH
equation, thereby extending the CH model to include space-charge
effects. On the other hand, CH2 has also been related to shallow
water dynamics by applying a series of approximations to the
Green-Naghdi equations \cite{CoIv2008}. Remarkably, a modified
version of CH2 dynamics has also been applied in image matching
\cite{HoTrYo2007}. (In image matching, the Hamiltonian is the norm
in which one applies optimal control.) Emergent peakon solutions
were found to result from applying $H^1$ smoothing to the
Hamiltonian in \cite{HoOnTr2009}.

The paper also identified several other potentially interesting open
problems. One of these is the problem of making physical
applications of the Kirillov ad$^*$-action
(\ref{Kirillov-ad-action}) for the Lie-Poisson bracket on the
symmetric Schouten algebra for arbitrary values of $(n,k)$. Another
is to determine whether the Klimontovich average in plasma kinetic
theory is a momentum map. One may also ask how the family of
symplectically conserved quantities corresponding to
statistical-moment versions of the Poincar\'e invariants found in
\cite{HoLySc1990} may fit into the theory of kinetic moments.

The kinetic approach used here may also provide physical interpretations of use in applying the BI system. For example, particle beams in linear accelerator lattices are described in terms of symplectic transfer matrices. (The same holds for linear ray optics.) In formulas, one has the relation $\bz(t)=\mathcal{M}(t)\,\bz(0)$, where $\mathcal{M}(t)$ is a one parameter subgroup of ${\rm Sp}(6,\mathbb{R})$ determining the beam evolution $\bz(t)\subset \mathbb{R}^{6}$. In this sense,
geodesics in ${\rm Sp}(6,\mathbb{R})$ would correspond to {\it
optimal transfer maps} for particle or optical beams. It is
interesting that similar approaches have recently emerged in quantum
computation, where the Hamiltonian of the system is constrained to
an optimal trajectory (i.e., a geodesic) by a cost function from
optimal control theory \cite{BrElHo08,NiDoGuDo08}. These additional open problems bode well for the potential success in future applications of using the geometric approach to Vlasov dynamics discussed here.  The application of ideas from optimal control to the Vlasov moments may be especially fruitful.

\subsection*{Acknowledgments} We are indebted with Fran\c{c}ois Gay-Balmaz, Simon Hochgerner, Boris Khesin, David Levermore, Cornelia
Vizman, Tudor Ratiu and Harvey Segur for useful and stimulating
discussions. The work of
DDH was also partially supported  by the Royal Society of London Wolfson Research Merit Award.

\end{document}